\documentclass[10pt,conference,compsocconf,letterpaper]{IEEEtran}
 \IEEEoverridecommandlockouts

\usepackage{graphicx}   
\usepackage{verbatim}   
\usepackage{color}  
\usepackage{subfig}

\usepackage{wrapfig}

\usepackage[ruled]{algorithm2e}
\usepackage[pointedenum]{paralist}
\usepackage{amssymb}
\usepackage{amsthm} 
\usepackage{amsmath}
\usepackage{hyperref}
\usepackage[all,2cell,ps]{xy} 
\usepackage{epstopdf}
\usepackage{multicol}
\usepackage{etex}
\usepackage{tikz}
\usetikzlibrary{arrows,backgrounds, matrix,positioning,fit,petri,topaths,shapes.misc, shapes.multipart,shapes,automata,graphs,calc,intersections,through,decorations.pathreplacing}
\usepackage{mdframed}

\theoremstyle{plain} \newtheorem{prop}{Proposition}
\theoremstyle{plain} \newtheorem{theorem}{Theorem}
\theoremstyle{plain} \newtheorem{property}{Property} 
\theoremstyle{plain} \newtheorem{lemma}{Lemma}
\newtheorem{definition}{Definition} 
\newtheorem{corollary}{Corollary} 
\theoremstyle{definition}

\renewcommand{\IEEEQED}{\IEEEQEDopen}

\ifCLASSINFOpdf
\else
\fi
\begin{document}
\vspace{-12mm}
%
\title{Efficient Polling Protocol  for  \\ Decentralized Social 
Networks$^*$\thanks{$^*$Funded by ANR Streams project.}\vspace{-10mm}}


\author{\IEEEauthorblockN{Hoang Bao Thien and Abdessamad Imine}
\IEEEauthorblockA{INRIA  Nancy -- Grand-Est, France\\
\{bao-thien.hoang, abdessamad.imine\}@inria.fr}
\vspace{-5mm}
}


%


\maketitle

\thispagestyle{plain}
\pagestyle{plain}

\begin{abstract}       
We address the polling problem in social networks where individuals collaborate to choose the most
favorite choice amongst some options, without divulging their vote and publicly exposing their potentially
malicious actions. Given this social interaction model, Guerraoui \textit{et al.}
\cite{DBLP:conf/opodis/GuerraouiHKM09,DBLP:journals/jpdc/GuerraouiHKMV12} recently proposed polling
protocols that do not rely on any central authority or cryptography system, using a simple secret sharing scheme
along with verification procedures to accurately compute the poll's final result.
However, these protocols can be deployed safely and efficiently provided that, inter alia, 
the social graph structure should be transformed into a ring structure-based overlay and
the number of participating users is perfect square. Consequently, designing \emph{secure} and \emph{efficient} polling
protocols regardless these constraints remains a challenging issue.

In this paper, we present EPol,
a simple decentralized polling protocol that relies on the current state of
social graphs. More explicitly, we define one family of social graphs that satisfy what we call the $m$-broadcasting property (where $m$ is less than or equal to the minimum node degree) and show
their structures enable low communication cost and constitute necessary and sufficient condition to ensure vote
privacy and limit the impact of dishonest users on the accuracy of the polling output. In a social network of $N$ users with diameter $\Delta_G$ and $D\le (m-1)\Delta_G/2$ dishonest ones (and similarly to the works \cite{DBLP:conf/opodis/GuerraouiHKM09,DBLP:journals/jpdc/GuerraouiHKMV12} where
they considered $D<\sqrt{N}$), a \textit{privacy parameter} $k$ enables us
to obtain the following results:
(i) the maximum probability that an honest node's vote is disclosed with certainty is
$(D/N)^{k+1}$ and without certainty is  $\bigl( \frac{D}{N}/(1-2\frac{D}{N})\bigr)\bigl[1-\sum_{i=0}^{k}\gamma_i(2\frac{D}{N})^{2i+1}\bigr]$ where $\gamma_i$ is the proportion of nodes voting for $2i+1$ shares and $0\le i\le k$;
(ii) up to $2D$ votes can be revealed with certainty;
(iii) the maximum impact on the final result is $(6k+4)D$, and the average impact is  $\Bigl[\bigl(\sum_{i=0}^{k}\gamma_i (2i+1)\bigr)\bigl(1+2\sum_{i=0}^{k}\gamma_i\frac{i+\alpha}{2i+1}\bigr)+1\Bigr]D$, where $\alpha$ is the proportion of users correctly voting;
(iv) unlike \cite{DBLP:conf/opodis/GuerraouiHKM09,DBLP:journals/jpdc/GuerraouiHKMV12}, our protocol is effective to compute more precisely the final result.
Furthermore, despite the use of richer social graph structures, the communication and spatial complexities of EPol are close to be linear.      
    
\end{abstract}


%
\IEEEpeerreviewmaketitle

\section{Introduction}\label{sec:intro}

In recent years, within billions of users, online social networks (OSNs) and social applications have changed the way users interact with the Internet. 
People can discuss, exchange photos and personal news, find others of a common interest, and many more.
\emph{Polling} is one of the current practical, useful but
important topics in OSNs.
In general, that is the problem of providing to all  participants  the outcome of a poll conducted among themselves, thus giving the most favorite choice among some options. 
Each participant can express his/her preference by submitting a vote, then all votes are aggregated and  the majority option will be chosen as the final result.
Just to demonstrate one typical example, a university has just launched a new administrative  procedure and may ask students whether or not this method is helpful, and each user will choose one option between ``Yes'' and ``No''.  
For the sake of simplicity, we here consider such a polling problem with only two options ``$+1$'' and ``$-1$'' for the concerning question. 



The goal in studying this problem is to devise a polling protocol such that it can perform a secure and accurate process to sum up the initial votes with the presence of dishonest users, who try to bias the outcome and disclose the votes of honest ones. 
Despite the simplicity characteristics of this problem, it takes an important part in incorporating 
user's opinion online. Thus, currently, several  studies and solutions for
this problem 
using two settings, centralized and distributed networks, are proposed.
In the centralized OSN, for instance Facebook Poll\footnote{http://apps.facebook.com/opinionpolls/} and  
Doodle\footnote{http://www.doodle.com/}, such a computation process may be easily achieved through a central server which is used to collect the users' votes before summing up them to obtain the output. 
Nevertheless, this solution suffers from server failures and particularly  privacy problems:  user might generally not want his/her vote to be known by a central entity, and it is not guaranteed the server will not bias and disclose the user votes.

\begin{table*}[t]
\begin{footnotesize}
\begin{center}
     \begin{tabular}{  l| r|  c|  c|  c|  c| c|  c  }
    \hline
    \textbf{Algorithm} & \textbf{Graph} & \textbf{Max. Impact} & \textbf{Privacy} & \textbf{Nb. of Dishonest Nodes} & \multicolumn{2}{ |c| }{\textbf{Complexity}} & \textbf{Crash}  	\\
      &   &    &  &   & \textbf{Spatial}  &  \textbf{Message} &      \\ \hline\hline

   \textbf{DPol} \cite{DBLP:conf/opodis/GuerraouiHKM09} &  Overlay & $(6k+2)D$ &  $(D/N)^{k+1}$ &  $D<\sqrt{N}$ &  $\mathcal{O}(rk+|g_i|)$ & $\mathcal{O}(rk+|g_i|)$ & Yes    \\   

   \textbf{DPol*} \cite{DBLP:journals/jpdc/GuerraouiHKMV12} &  Overlay & $(6k+4)D$ &  $(D/N)^{k+1}$ &  $D<\sqrt{N}$ &  $\mathcal{O}(rk+|g_i|)$ & $\mathcal{O}(rk+|g_i|)$ & Yes    \\   
                  
   \textbf{Pol} \cite{DBLP:conf/opodis/HoangI12}  &  General & $(6k+4)D$ &  $(D/N)^{k+1}$ &  $D<N/5$ &  $\mathcal{O}(N^2)$ & $\mathcal{O}(k+N^2)$ & No   \\   

   \textbf{MPOL*} \cite{DBLP:conf/trustcom/EnglertG13}  &  Overlay & $(6k+2)D$ &  $(D/N)^{k+1}$ & $D<\sqrt{N}$  &  $\mathcal{O}(rk^2|g_i|)$ & $\mathcal{O}(rk+|g_i|)$ & No   \\   

   \textbf{PDP} \cite{DBLP:conf/srds/BenkaouzGEH13} &  Overlay & $2(k+N)D$ &  $(D/N)^{k+1}$ &  $D\ge k+1$ &  $\mathcal{O}(rk+|g_i|)$ & $\mathcal{O}(rk+|g_i|)$ & No   \\   
                  
   \textbf{DiPA} \cite{DBLP:journals/computing/BenkaouzE13}  &  Overlay & $2(k+N)D$ &  $(D/N)^{k+1}$ & $D\ge k+1$ &  $\mathcal{O}(rk+|g_i|)$ & $\mathcal{O}(rk+|g_i|)$ & No   \\   
                  
   \textbf{This work}   &  General & $(6k+4)D$ &  $(D/N)^{k+1}$ &  $D\le(m-1)\Delta_G/2$ &  $\mathcal{O}(k+mN)$ & $\mathcal{O}(k+N.(d_0-m))$ & Yes   \\   

  \hline
    \end{tabular}
\end{center}
\caption{\small Comparison of distributed polling protocols where ``Max. Impact'': the maximum difference between the output and the expected result, ``Privacy'': the probability a node's vote is revealed, ``Nb. of Dishonest Nodes'': the number of dishonest nodes the system can tolerate, ``Spatial complexity'': the total space a node must hold,  ``Message complexity'': the number of messages sent by a node, $r$: number of groups, $|g_i|$: group size, $d_0$: maximum node degree, $\Delta_G$: network diameter. Entries marked with an asterisk (*) show the results for binary polling.} 
\label{table:comparision_protocols}    

\end{footnotesize}
\vspace{-5mm}
\end{table*}

In this paper, we focus on the polling protocol based on the decentralized OSN where
information is not concentrated in a central point and hence, user privacy is improved. 
In addition, we do not want to rely on cryptography for ensuring privacy or accuracy because: (i) the cryptography uses complicated computation that impacts to the scalability and practicality of the protocols; (ii) all cryptography
techniques rely on the assumptions which are not proved and might be broken (e.g., the
difficulty for factorizing the product of two big prime numbers or solving the discrete logarithm problem); and (iii) some traditional distributed computing problems  can be solved without cryptography as motivated in \cite{DBLP:conf/fc/MalkhiMP02,DBLP:conf/rsalabo/rivest98}. 

Guerraoui  \textit{et al.} \cite{DBLP:conf/opodis/GuerraouiHKM09,DBLP:journals/jpdc/GuerraouiHKMV12}
have opened the way to a novel and promising approach to perform secure distributed polling by proposing DPol, a simple distributed polling protocol  based on secret sharing scheme and without using cryptography. 
The authors also presented the verification procedures to dissuade user misbehaviors, and enable honest users to tag profiles of dishonest ones.
The probability of dishonest users violating privacy is balanced by their impact on the accuracy of the final result.   
However, DPol has practically some disadvantages. Firstly, DPol relies on a
structured overlay, a cluster-ring-based structure inspired from \cite{DBLP:journals/ipl/GalilY87}. 
Although it is efficient in term of communication cost, it  is on top and really apart  from the normal social graph.
It does not take 
into account the social links among users in the sense  it builds the uniform distribution of users into groups. This is not practical as
we have to target a special case using
notion of group instead of reserving the normal structure of the graph. 
Secondly,
the number of users should be a perfect square number such that one social graph with 
$N$ users 
is divided into $\sqrt{N}$ groups of size $\sqrt{N}$.                            
Thirdly, DPol assumes that dishonest nodes cannot wrongfully blame honest ones which seems to be a too strong assumption. This assumption helps honest nodes to correctly identify the dishonest misbehaviors.
Finally, transforming a graph into an overlay is not always useful since that may affect some security properties like accuracy and complexity.
For instance, if the network is a clique, then each node can easily obtain all data from all its friends because they are fully connected. However after transforming a clique into an overlay,  as  the data sent by a node may be corrupted by the intermediate nodes, to ensure the accuracy,  DPol should do a verification procedure which of course increases the communication cost. In addition, as stated in DPol,
an honest node may decide on the arbitrary data sent by dishonest ones, and thus, it might be incorrect. It implies the impact of dishonest nodes on accuracy may be high.

Later, several protocols and extensions inspired from the idea of DPol have been proposed such as
MPOL \cite{DBLP:conf/trustcom/EnglertG13}, PDP \cite{DBLP:conf/srds/BenkaouzGEH13} and DiPA \cite{DBLP:journals/computing/BenkaouzE13}. 
However these protocols rely on the same ring-based overlay structure, and have minor contribution compared to DPol. 
Unlike these works, authors of \cite{DBLP:conf/opodis/HoangI12} introduced a distributed polling protocol and a family of more general social graphs  which ensures the correctness of the protocol and vote privacy of nodes. 
Nonetheless, the communication model is \emph{synchronous} and the communication cost is super-linear in $N$, and is $\mathcal{O}(N^2)$ in the worst case, with the presence of dishonest nodes. 
It should be noted, as opposed to \cite{DBLP:conf/opodis/HoangI12}, all 
other works based on an overlay structure considered complexities with only honest nodes.

Accordingly, devising efficient and decentralized polling protocols without cryptography
and constraints (such as the use of  overlay structure and perfect square number of users) imposed
in \cite{DBLP:conf/opodis/GuerraouiHKM09,DBLP:journals/jpdc/GuerraouiHKMV12}
remains a challenging problem. 

\noindent\textbf{Contributions.}
In this paper, we propose the design of a simple decentralized polling protocol not requiring any central authority or cryptography system. Unlike
\cite{DBLP:conf/opodis/GuerraouiHKM09,DBLP:journals/jpdc/GuerraouiHKMV12}, our protocol is deployed on the original social network in such a way 
each individual can perform the voting process privately and securely without resorting to the group division.
We describe properties required for the social graph to ensure the correctness of our protocol. 
Despite the use of richer social graph structures (they also include the ring-based structure given in
\cite{DBLP:conf/opodis/GuerraouiHKM09,DBLP:journals/jpdc/GuerraouiHKMV12}), 
one node can receive/send so many duplicated  messages from/to other nodes.  This can lead to flooding the local storage and getting high communication cost. 
Inspired from \cite{DBLP:conf/isit/ShahRR13}, we introduce a method for efficiently broadcasting messages by using the concept that we call  the \emph{$m$-broadcasting property}.
A graph satisfies the $m$-broadcasting property for a parameter $m\in \mathbb{N}$ such that  $1\le m\le d_{min}$, where $d_{min}$ is the minimum node degree, if for each source node, there exists a topological ordering of the nodes such that every node either connects directly to the source or to some $m$ nodes preceding it in the ordering w.r.t. the source. Consequently, instead of accepting all messages originating from a source, a node stores only $m$ ones passed by \emph{ordered paths}.
Appendix \ref{appxsub:particular_graphs} presents some examples of social graphs satisfying the $m$-broadcasting property.
The construction of a graph satisfying the $m$-broadcasting property from a general graph is \emph{beyond} the scope of this paper. However, to be more practice, we describe an algorithm to construct  this kind of graphs in   Appendix \ref{appxsub:algo_construct_m_broad}. 

To describe carefully the distributed implementation of a polling problem, we consider  the following fundamental criteria:
accuracy, privacy, 
resilience to dishonest nodes,
and asymptotic complexity. 
Using the same notion of privacy parameter $k$ in \cite{DBLP:conf/opodis/GuerraouiHKM09,DBLP:journals/jpdc/GuerraouiHKMV12}, we get the following results in a system of size $N$ with $D$ dishonest users:
(i) the probability that an honest node's vote is disclosed with certainty is at most
$(D/N)^{k+1}$;
(ii) up to $2D$ votes can be revealed with certainty by the dishonest  coalition;
(iii) in practice, dishonest nodes may also try to reveal a node's vote even if they only get some partial information of the vote (e.g., some shares of that vote). 
We consider the case where dishonest nodes agree on some uncertain vote disclosure rules (section \ref{sec:protocol_with_dishonest}), the maximum probability of greedy (i.e., analyzing \emph{some} shares of the vote) and non-greedy (i.e., analyzing \emph{all} shares of the vote) vote detection are respectively   $\sum_{i=\rho}^{k}\gamma_i\frac{N+1}{N-D+\rho+2}  (\frac{D}{N-D+\rho+1}  ) ^{\rho+1}$
and $\bigl( \frac{D}{N}/(1-2\frac{D}{N})\bigr)\bigl[{1-\sum_{i=0}^{k}\gamma_i(2\frac{D}{N})^{2i+1}}\bigr]$
where $\gamma_i$ is the proportion of nodes voting with $2i+1$ shares and $0\le\rho\le i\le k$;
(iv) the maximum impact  from the dishonest coalition to the final result is $(6k+4)D$, and the average impact is  $\Bigl[\bigl[\sum_{i=0}^{k}\gamma_i (2i+1)\bigr]\bigl[1+2\sum_{i=0}^{k}\gamma_i\frac{i+\alpha}{2i+1}\bigr]+1\Bigr]D$, where $\alpha$ is the proportion of users correctly voting; 
(v) the maximum number of dishonest nodes that the system can tolerate is $(m-1)\Delta_G/2$, where $\Delta_G$ is the network diameter; 
and (vi) the communication and spatial complexities of our protocol are close to be linear.

We are aware  that the data sent by a node may be corrupted by the intermediate dishonest nodes. Thus, an honest node
may receive distinct values of the same data. 
As opposed to DPol \cite{DBLP:conf/opodis/GuerraouiHKM09,DBLP:journals/jpdc/GuerraouiHKMV12}, we ensure each node can decide the most represented value to obtain \emph{correct} data of other ones.

In addition, 
most of the previous works assume the existence of reliable communication among nodes. However, 
nodes communicate by UDP which may suffer message loss on the communication channels or nodes' crash. In this work, we analyze the effect of these factors  on accuracy and termination of the protocol by determining the impact on the final outcome and the probability of a node failing to decide and compute the final result. 

We illustrate the comparison of contributions between our work and other distributed polling protocols in Table \ref{table:comparision_protocols}.
This table shows that our protocol \emph{tolerates more} dishonest nodes and has \emph{better complexities} than other ones.
More particularly, if the graph is a ring structured, compared to DPol \cite{DBLP:conf/opodis/GuerraouiHKM09,DBLP:journals/jpdc/GuerraouiHKMV12}, our protocol
has the same message complexity, but   tolerates more dishonest nodes and computes more accurately the  poll outcome.
It is also noted that DPol 
investigated  the effect  of crash only, and not of message loss.

\noindent\textbf{Outline.} 
This paper is organized as follows. Section \ref{sec:problem_state} describes
our
polling model, and introduces a family of social graphs. Section
\ref{polling_protocol}
presents our polling protocol and Section \ref{sec:correctness}
analyzes its correctness  
with and without the presence of dishonest nodes.
Section \ref{sec:crash_msg_loss} discusses the impact of crash and message loss on accuracy and termination of the protocol. We review related work in Section
\ref{sec:related_work} and conclude the paper with future research
 in Section \ref{sec:conclusion}.
\vspace{-2mm}

\section{Social network model}\label{sec:problem_state} 
This section defines the user behaviors and presents
the social graph models. It should be noted that we
consider the same assumptions given in
\cite{DBLP:conf/opodis/GuerraouiHKM09,DBLP:journals/jpdc/GuerraouiHKMV12}.
\vspace{-3mm}

\subsection{User behaviors}\label{subsec:system_model}
The polling  problem consists of a system 
modeled as the form of an undirected social  graph $G=(V,E)$  
with $N=|V|$ uniquely  identified nodes representing users and $E$ is a set of social links. 
Each participant $n$ expresses its opinion by giving a vote $v_n\in \{-1,1\} $. After collecting the votes of all nodes, the expected
outcome is $\sum_{\substack{n}}v_n $. 
In this work, we consider the following assumptions:

We consider the \emph{asynchronous} model where each node can communicate (send/receive messages) with its neighbors (e.g., direct friends).
Some messages may arrive to the destination with some delay.
All nodes have to send/receive/forward messages  if they are requested.

Each node is either honest or dishonest.  The honest node completely complies with the protocol and takes care about its privacy in the sense that  the vote value is not  disclosed. 
All nodes care about their \emph{reputation}:  information related to a user is intimately considered to reflect the associated real person.  
In particular, dishonest nodes never do any misbehavior which will jeopardize their reputation. 

All dishonest nodes can form a coalition to get the full  knowledge of the network and try
to do everything to achieve these goals without being detected: (i) bias the result of the election by promoting their votes or changing the values they received from other honest nodes;
(ii) infer the opinions of other nodes. 

In order to unify the opinions and not give compensating effects, all dishonest nodes make the single coalition $\mathcal{D}$ of size $D$. However, they also want  to protect their reputation from being affected. They are selfish in the sense that each dishonest node prefers to take care about its own reputation to covering up each other \cite{DBLP:conf/opodis/GuerraouiHKM09,DBLP:journals/jpdc/GuerraouiHKMV12}. As such the  dishonest nodes  are rather restricted but more reflective of the real human behavior than Byzantine nodes \cite{DBLP:journals/toplas/LamportSP82}.

To tolerate the existence of dishonest nodes with a limited vote corruption, we assume each node has at least one honest friend but it does not know which friend is honest or not.


To dissuade the user misbehaviors, an activity affected to the profile of the concerned node is given. More precisely, if node $u$ is detected
as dishonest one by $v$ then $u$'s profile is tagged with the statement
``$v$ accused $u$ as a dishonest user'' and $v$'s profile has the
statement ``$u$ is a bad guy''. Notice that in social networks, no one would like to be tagged as dishonest. Furthermore, we do not take into account 
the Sybil attacks  and
spam since those kinds of misbehaviors can be detected by some tools or several
 systems such as SybilGuard\cite{DBLP:journals/ton/YuKGF08}, SybilLimit
\cite{DBLP:journals/ton/YuGKX10},
and
\cite{DBLP:conf/nsdi/MislovePDG08,DBLP:conf/infocom/SirivianosKY11} 
(for mitigating spam).
However,
dishonest nodes can wrongfully blame other ones.
\vspace{-3mm}

\begin{figure}[t]
\centering
\vspace{-10mm}
\begin{tikzpicture} 
[>=stealth', thick, scale=0.8,  
terminal/.style={
circle,
minimum size=3mm,
very thick,draw,
font=\sffamily}, 
every label/.style={thick}
]
\foreach \name/\pos/\value in {{n/(0,0)/n},{u1/(-2,1)/x},{u2/(-2,0)/y},{u3/(-2,-1)/z},{v1/(2,1)/a},{v2/(2,0)/b},{v3/(2,-1)/c}}
	\node[terminal,label=$\value$] (\name) at \pos {}  ;

\foreach \source/\dest in {n/v3,n/v2,u1/n,u2/n}                                   
	\path[->] (\source) edge node[right] {} (\dest) ;
\foreach \source/\dest  in {n/v1}                                  
	\path[] (\source) edge (\dest) ;                  
\foreach \source/\dest in {n/u3}                                   
	\path[<->] (\source) edge node[right] {} (\dest) ;  		

\draw [decorate,decoration={brace,amplitude=10pt,mirror} ]
(u1.north -| u1.west)--(u3.south -| u3.west) node [black,midway,xshift=-40pt] {\footnotesize producers $\mathcal{R}_n$}; 

\draw [decorate,decoration={brace,amplitude=8pt} ]
(v2.north -| v2.east)--(v3.south -| v3.east) node [black,midway,xshift=40 pt] {\footnotesize consumers $\mathcal{S}_n$};

\end{tikzpicture} 
\caption{Producers and consumers of $n$.}
\label{fig:producers_vs_consumers}
\vspace{-7mm}
 
\end{figure}
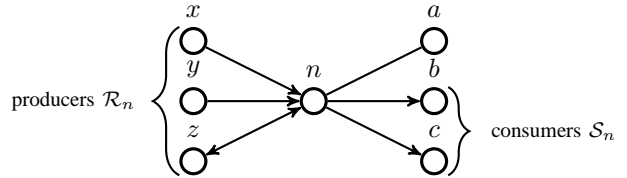

\subsection{Social graph model}\label{sec:social_graph_model}
In this section, firstly, we define the terms and notations of the social graph used
throughout this work. Then we describe the graphs with low communication cost. Finally, we demonstrate the description of our family of graphs.
\vspace{-1mm}

\noindent{\textbf{Notations.}} \label{par:notation}
Each node $n$  maintains a set of direct neighbors $\Gamma(n)$ of size $d_n$, and two subsets of $\Gamma(n)$: 
a set $\mathcal{S}_n$ of \emph{consumers} containing nodes that $n$ sends messages to, and $\mathcal{R}_n$ of \emph{producers} relating to nodes for which $n$ acts as a consumer. They might not be  disjoint, i.e., $\mathcal{S}_n\cap\mathcal{R}_n\neq \varnothing$, as depicted in Fig. \ref{fig:producers_vs_consumers} (where node $z$ is a producer and is a consumer as well). 
The procedure to define the set of producers and consumers is as follows. Each node chooses an odd number $2i+1$ of neighbors  as its consumers such that $0\le i\le k$ where $k$ is a privacy parameter. It then establishes a three-way-handshake protocol to inform its consumers that it is one of their producers. By this way, it also obtains a set of producers.

We denote $\Delta_G$ as the diameter of the network $G$.
To our best knowledge, the distributed algorithms for computing exact  diameter  take time $\mathcal{O}(N)$ \cite{DBLP:conf/podc/HolzerW12,DBLP:conf/icalp/PelegRT12}.

\noindent{\textbf{Graphs with low communication cost.}}
In the network, consider the broadcasting operation initiated by a single node, called  \emph{source}. The source has a data and wants to disseminate it to all nodes in the network. 
In the naive approach, upon receiving a message from a neighbor, a node stores the data then forwards it on every other edge. Despite the use of richer social graph structures, one node can receive/send so many duplicated messages (which may be passed by many paths) from/to other nodes. This leads to flooding the local storage.
In this paper, inspired from \cite{DBLP:conf/isit/ShahRR13}, we propose a method for efficiently broadcasting message by using the concept that we call  the \emph{$m$-broadcasting property}.
A graph satisfies the $m$-broadcasting property for a positive integer $m$ such that  $1\le m\le d_{min}$, where $d_{min}$ is the minimum node degree, if for each source node, there exists a topological ordering of the nodes in the graph such that every node either connects directly to the source or to some $m$ nodes preceding it in the ordering w.r.t. the source. Accordingly, instead of accepting all messages originating from a source, a node receives and stores only $m$ ones passed by ordered paths.

We denote by $\beta_n(s)$ a number of neighbors of $n$ preceding it in the ordering w.r.t. source $s$. (We sometimes omit the mentioned source where no confusion arises.)
\vspace{-2mm}

\subsection{Secret sharing based graphs}\label{subsec:graph_model}

In this part, we present the family of graphs without and with the
presence of dishonest nodes.
We use a predefined parameter $k\in \mathbb{N}$ (like
\cite{DBLP:conf/opodis/GuerraouiHKM09,DBLP:journals/jpdc/GuerraouiHKMV12}) and a parameter $m\in \mathbb{N}$ to present the features of our
social graphs. Let $G=(V,E)$ be a social graph with the following properties:
\begin{property}[$P_{g_1}$] \label{graph_prop1}
$d_n\ge 2k+1$, $|\mathcal{S}_n|=2i+1$ and $|\mathcal{R}_n|\le 2k+1$ where $0\le i\le k$, for every $n\in V$.
\end{property} 

\begin{property}[$P_{g_2}$] \label{graph_prop2}
$G$ satisfies the $m$-broadcasting property.
\end{property}
 
\begin{property}[$P_{g_3}$] \label{graph_prop3}
For a source node, each other node has less than $m/2$ dishonest neighbors preceding it in the ordering (w.r.t. source node).
\vspace{-2mm}
\end{property}

According to Property $P_{g_1}$, the set of consumers and the set of producers of a node have the size of \emph{at most} $2k+1$ and might \emph{not be disjoint}. This condition distinguishes our graph family from other structures in \cite{DBLP:conf/trustcom/EnglertG13,DBLP:conf/opodis/GuerraouiHKM09,DBLP:journals/jpdc/GuerraouiHKMV12} and is more flexible than a graph family in \cite{DBLP:conf/opodis/HoangI12} since they all consider the restricted condition where each node has \emph{exactly} $2k+1$ consumers. In addition, it also differs from \cite{DBLP:journals/computing/BenkaouzE13,DBLP:conf/srds/BenkaouzGEH13} which do not give any condition to the upper bound of the number of producers (that one node should have), thus, a dishonest node can send arbitrary summing data to others and the accuracy of the global outcome is easily affected. 
Property $P_{g_2}$ enables us to reduce the communication cost   in the system.
It is also noted that this condition implicitly implies the condition that $G$ is an honest graph mentioned in \cite{DBLP:conf/opodis/HoangI12}, i.e.,  for every honest nodes $u$, $v$, there exists a path between $u$ and $v$ containing only intermediate honest nodes.
Property $P_{g_3}$ ensures each honest node always obtains one correct version of data from other honest ones. 
 Property  $P_{g_3}$ also enables us to limit the size of dishonest users, that is $D\le \frac{m-1}{2}\Delta_G$ (presented in Theorem \ref{theo:max_tolerance}).

\noindent From these properties, we characterize two families of graphs:
\begin{enumerate}[(i)]
\vspace{-1mm}
 \item $\mathcal{G}_1=\{G\mid \mathcal{D}(G)=\varnothing$ and $G$ satisfies
$P_{g_1}, P_{g_2}\}$.
\item $\mathcal{G}_2=\{G\mid \mathcal{D}(G)\ne \varnothing\empty$ and $G$
satisfies
$P_{g_1}, P_{g_2} \mbox{ and } P_{g_3}\}$. 

\end{enumerate}
\vspace{-1mm}

\noindent Graphs in $\mathcal{G}_1$ contain only honest nodes and satisfy property $P_{g_1}$ and $P_{g_2}$. Graphs in $\mathcal{G}_2$ contain honest and dishonest nodes and satisfy properties $P_{g_1}$, $P_{g_2}$ and $P_{g_3}$.\vspace{-2mm}

\SetKwData{result}{result}
\begin{algorithm*}[t!]

\begin{scriptsize}
\DontPrintSemicolon

\caption{\sc Polling algorithm at node $n\in
\{0,1,...,N-1\}$}\label{algo_polling}

\begin{multicols}{2}

\SetKwFunction{Share}{Sharing}

\SetKwFunction{Broadcast}{Broadcasting}

\SetKwFunction{Aggregate}{Aggregating}
\SetKwFunction{Decide}{Decide}
\SetKwFunction{EventReceiveShare}{RecShareEvt}
\SetKwFunction{EventReceiveData}{RecDataEvt}
\SetKwFor{Upon}{upon event}{do}{endupon}

\SetKwInOut{Input}{Input}\SetKwInOut{Output}{Output}

\SetKwInOut{Var}{Variables}

\textbf{Input:} 

 \begin{tabular}{rl}
 
 $v_n$: & A vote of node $n$, value in $ \{-1,1\}$\\
 $k$~: & privacy parameter \\
 $m$~: & positive integer where $1\le m \le d_{min}$ \\
 \end{tabular}
\BlankLine
\textbf{Variables:}

\begin{tabular}{rl}
  $c_n$:& collected data, $c_n=0$\\
  $\mathcal{C}_n$:&  set of possible collected data\\
  & $\mathcal{C}_n[\{0,1,...,N-1\}\rightarrow \varnothing]$ \\
  $h_n$: & set of final deciding collected data \\
  & $h_n[\{0,1,...,N-1\}\rightarrow  \bot ]$ \\
\end{tabular}
\BlankLine
\textbf{Output:} \result

\noindent\rule{.46\textwidth}{0.4pt}
\textbf{Polling Algorithm} 
\noindent\rule[1mm]{.46\textwidth}{0.4pt}

\SetAlgoVlined
\DontPrintSemicolon

\nl $\Share(v_n,\mathcal{S}_n,i)$   \;
\nl $\Broadcast(c_n)$\;  
\nl $\Aggregate() $


\noindent\rule{.46\textwidth}{0.4pt}
\textbf{Procedure} $\Share(v_n,\mathcal{S}_n,i)$
\noindent\rule[1mm]{.46\textwidth}{0.4pt}


%

\DontPrintSemicolon

\SetKwData{send}{send}

	 \nl $\mathcal{P}_n \leftarrow \{v_n\}$ \label{mark1str_vot_proc}\;
	 \nl \For{$j \leftarrow 1 $ \KwTo $i$ }{
 	    \nl $\mathcal{P}_n \leftarrow \mathcal{P}_n \cup \{v_n\}\cup\{-v_n\}$\;
	 }  
	 \nl $\mu_n\leftarrow_{rand} \mathcal{P}_n$ 
	 \BlankLine
	 \nl \For{$j \leftarrow 0$ \KwTo $2i$ }{
	  \nl $\send$ (SHARE, $\mu_n[j]$) to  $\mathcal{S}_n[j]$ \label{mark1end_vot_proc}
		}


\BlankLine

\nl $count\leftarrow 0$ \label{mark2str_vot_proc}\;
\nl \While{\emph{(}$count<|\mathcal{R}_n|$\emph{)} }{


  \nl \Upon{\emph{(receive (SHARE, $p$) from neighbor $r$ in the first time)}}{
  \nl \If{$(r\in \mathcal{R}_n$ \textbf{\emph{and}} $p\in \{-1,1\})$}{
  \nl $c_n \leftarrow c_n + p$ \;
  \nl $count\leftarrow count+1$ \label{mark2end_vot_proc}\;
    }           
  }
}

\BlankLine
\columnbreak    

\noindent\rule{.475\textwidth}{0.4pt}
\textbf{Procedure} $\Broadcast(c_n)$
\noindent\rule[1mm]{.475\textwidth}{0.4pt}

%

\SetKwData{True}{true}
\SetKwData{False}{false}
\SetKwData{Exit}{exit}
\SetKwData{Break}{break}
\SetKwData{Continue}{continue}

  \nl 
    \ForEach{$(r\in \Gamma(n))$\label{mark1st_br_proc}}
	   {\nl \send(DATA, $n$, $c_n$) to $r$  \label{mark1end_br_proc}\;}
  \nl $count\leftarrow 0$ \;
  \nl \While{\emph{(}$count<N-1$\emph{)} }{

   \nl \Upon{\emph{(receive message (DATA, $s$, $c_s$) from direct neighbor $r$ preceding $n$ in the ordering w.r.t. source $s$)}}{
		\nl \uIf{$(r=s)$  \label{mark2st_br_proc}}
			   { \nl $h_n[s]\leftarrow c_s$\;
			     \nl $count\leftarrow count+1$\;
			     \nl Forward (DATA, $s$, $c_s$) to other friends succeeding $n$ in the ordering w.r.t. source $s$ \label{mark2end_br_proc}
			    }
		  \nl \ElseIf{$(s\notin \Gamma(n))$}{  
		        \nl $\mathcal{C}_n[s]\leftarrow \mathcal{C}_n[s]\cup \{c_s\}$  \label{mark3st_br_proc}\;
			 \nl \If{$(|\mathcal{C}_n[s]|=m)$  \label{mark4st_br_proc}}{
			   \nl $h_n[s]\leftarrow \Decide(\mathcal{C}_n[s])$\;
			   \nl $count\leftarrow count+1$\;
			   \nl \send (DATA, $s$, $h_n[s]$) to other $d_n-\beta_n(s)$ friends succeeding $n$ in the ordering w.r.t. source $s$  \label{mark4end_br_proc}
                }
	 
		  }
	   }
  }
\BlankLine

\noindent\rule{.475\textwidth}{0.4pt}
\textbf{Function} $\Decide(\mathcal{Z}) $
\noindent\rule[1mm]{.475\textwidth}{0.4pt}

%
%

  \nl \Return the most represented value in $\mathcal{Z}$ 

\noindent\rule{.475\textwidth}{0.4pt}
\textbf{Procedure} $\Aggregate() $
\noindent\rule[1mm]{.475\textwidth}{0.4pt}

%
%
%
  \nl $\result\leftarrow 0$  \label{mark1st_aggr_proc}\;
  \nl \For{$s \leftarrow 0$ \KwTo $N-1$}{
  \nl \uIf{ $(s\ne n)$ }{
	  \nl $\result \leftarrow  \result +h_n[s]$ \;
	 }\nl\lElse { $\result \leftarrow  \result +c_n$ \label{mark1end_aggr_proc}}
  }  
\end{multicols}
\end{scriptsize} \vspace{-2mm}  
\end{algorithm*} 
\vspace{-1mm}


\section{Polling protocol}\label{polling_protocol}

In this section, we present our polling protocol, \emph{EPol},  for the network without crash and message loss. EPol is composed of the following  phases:


\noindent\textbf{Sharing.} In this phase, each node $n$ contributes its
opinion
by sending a set of shares expressing its vote $v_n\in \{-1,1\}$ to its
consumers.
We inspired the sharing scheme proposed in 
\cite{DBLP:conf/opodis/Delporte-GalletFGR07} to generate shares. Namely, first $n$ chooses \emph{randomly} a value $i$ such that $i\in \{0,1,...,k\}$. This value $i$ is not known by other nodes. 
Then it
generates $2i+1$ shares  $\mathcal{P}_n=\{p_1,p_2,...,p_{2i+1}\}$, where  
$p_j\in\{-1,1\}$ and $1\le j\le 2i+1$,  including: $i+1$ shares of value $v_n$,
and $i$ shares of opposite $v_n$'s value. 
The intuition of this creation is to regenerate the vote $v_n$ when the shares are summed. 
Later it randomly generates
 a permutation $\mu_n$ of $\mathcal{P}_n$, and sends shares to
$2i+1$ consumers. Lines \ref{mark1str_vot_proc}--\ref{mark1end_vot_proc} in Algorithm \ref{algo_polling} describe this activity. 
Node also receives  $|\mathcal{R}_n|$ shares from its producers. 
Note that $\mathcal{S}_n$ and $\mathcal{R}_n$ might not be disjoint.\footnote{This distinguishes our protocol from approaches in \cite{DBLP:journals/computing/BenkaouzE13,DBLP:conf/srds/BenkaouzGEH13,DBLP:conf/trustcom/EnglertG13,DBLP:conf/opodis/GuerraouiHKM09,DBLP:journals/jpdc/GuerraouiHKMV12}. The set of consumers and producers in these approaches are separated for each of size $2k+1$.}

After each node collects   shares from its producers, and sums into
\textit{collected data} $c_n$ (lines \ref{mark2str_vot_proc}--\ref{mark2end_vot_proc}  in Algorithm \ref{algo_polling}),
this phase is over.
It is also noted that because the votes and their generating shares belong to the set $\{-1,+1\}$, nodes cannot distinguish between a vote and a share. Hence,  if a node opts a value $i=0$ and generates only $2i+1=1$ share, the dishonest consumer receiving a message from that node could not distinguish if such share was generated as a single one or it is one among  many generated shares of that node.

Figure \ref{protocol_figure} illustrates an example of the protocol for $i=k=1$. 
Figure \ref{fig:protocol_sf1} presents the network and the ordering of nodes w.r.t. source $A$ in the parentheses. Figure \ref{fig:protocol_sf2} depicts the sharing phase at node $A$. Node $A$ would like to vote $+1$, thus, it generates  a set of $2i+1=3$
shares $\{+1,-1,+1\}$. 
Node $A$ collects the shares from its producers and computes the collected
data $c_A=3$. 
\vspace{-1mm}

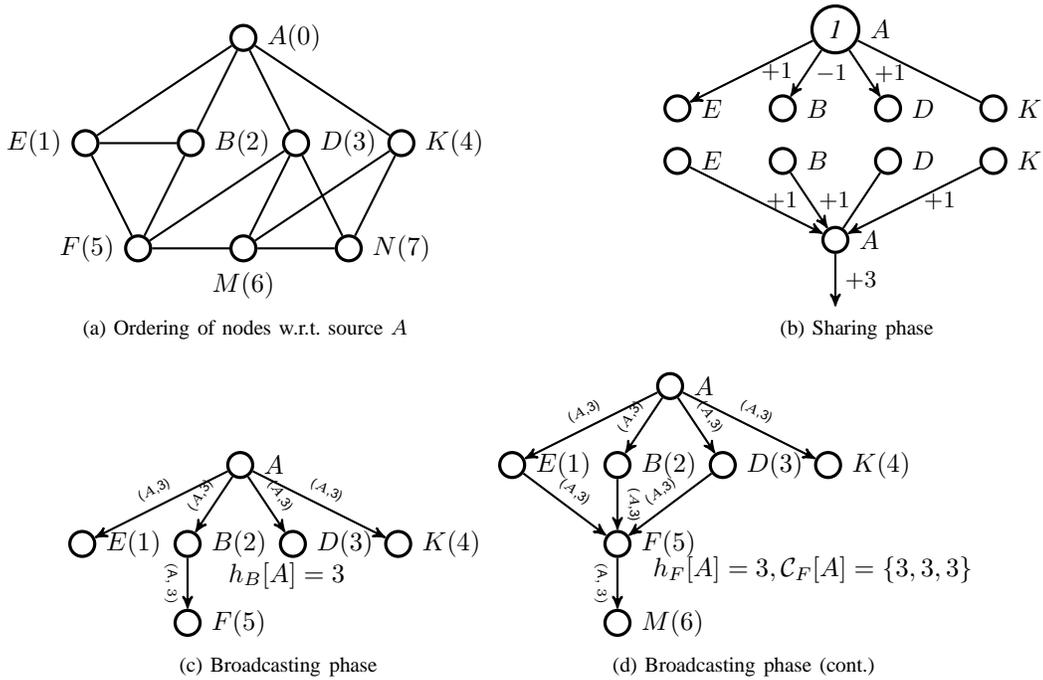
\begin{figure*}[t]
\vspace{-10mm}
\centering
\hspace*{\fill}
\subfloat[Ordering of nodes w.r.t. source $A$]{
\begin{tikzpicture} 
[>=stealth', thick, scale=0.7,  
nonterminal/.style={
	circle,
	minimum size=3mm,
	very thick,
	draw,
	font=\itshape
},
terminal/.style={
circle,
minimum size=3mm,
very thick,draw,
font=\sffamily}, 
every label/.style={thick}
]

\foreach \name/\value/\pos in {{A/$A(0)$/(0,0)}, {B/$B(2)$/(-1,-2)},{D/$D(3)$/(1,-2)},{K/$K(4)$/(3,-2)}}
	\node[terminal,label={right:\value}] (\name) at \pos {}  ;

\foreach \name/\value/\pos in {{M/$M(6)$/(0,-4)}}
	\node[terminal,label={below:\value}] (\name) at \pos {}  ;

\foreach \name/\value/\pos in {{F/$F(5)$/(-2,-4)},{E/$E(1)$/(-3,-2)}}
	\node[terminal,label={left:\value}] (\name) at \pos {}  ;

\foreach \name/\value/\pos in {{N/$N(7)$/(2,-4)}}
	\node[terminal,label={right:\value}] (\name) at \pos {}  ;

\foreach \source/\dest  in {
A/E,A/B,A/D,A/K,E/B,E/F,D/M,D/N,K/N,B/F,D/F,F/M,M/N,M/K}                                   
	\path[-] (\source) edge node[font=\sffamily\small,auto,swap] {} (\dest) ;

\end{tikzpicture}
\label{fig:protocol_sf1}
}
\hfill
\subfloat[Sharing phase]{
\begin{tikzpicture} 
[>=stealth', thick, scale=0.7,  
nonterminal/.style={
	circle,
	minimum size=3mm,
	very thick,
	draw,
	font=\itshape
},
terminal/.style={
circle,
minimum size=3mm,
very thick,draw,
font=\sffamily}, 
every label/.style={thick}
]
 
\foreach \name/\pos/\value in {{E/(-3,-1)/},{B/(-1,-1)/},{D/(1,-1)/},{K/(3,-1)/}}
	\node[terminal,label=right:$\name$] (\name) at \pos {\value}  ;

\foreach \name/\pos/\value in {{A/(0,0.5)/1}}
	\node[nonterminal,label=right:$\name$] (\name) at \pos {\value}  ;

 \foreach \source/\dest/\text  in {A/E/+1,A/B/-1,A/D/+1}                                   
	\path[->] (\source) edge node[right,font=\sffamily\small] {$\text$} (\dest) ;
\foreach \source/\dest  in {A/K}                                  
	\path[-] (\source) edge (\dest) ;


\foreach \name/\pos/\value in {{E1/(-3,-2)/E},{B1/(-1,-2)/B},{D1/(1,-2)/D},{K1/(3,-2)/K}}
	\node[terminal,label=right:$\value$] (\name) at \pos {}  ;

\foreach \name/\pos/\value in {{A1/(0,-3.5)/A}}
	\node[nonterminal,label=right:$\value$] (\name) at \pos {}  ;
 \coordinate [below=0.7cm of A1] (x);

\foreach \source/\dest/\text  in {E1/A1/+1,B1/A1/+1,K1/A1/+1,A1/x/+3}                                   
	\path[->] (\source) edge node[right,style={,thick},font=\sffamily\small] {$\text$} (\dest) ;
\foreach \source/\dest  in {A1/D1}                                  
	\path[] (\source) edge (\dest) ;
%
\end{tikzpicture}
\label{fig:protocol_sf2}
}
\hspace*{\fill}
\\
\hspace*{\fill}
\subfloat[Broadcasting phase]{
\begin{tikzpicture} 
[>=stealth', thick, scale=0.7,  
nonterminal/.style={
	circle,
	minimum size=3mm,
	very thick,
	draw,
	font=\itshape
},
terminal/.style={
circle,
minimum size=3mm,
very thick,draw,
font=\sffamily}, 
every label/.style={thick}
]

\foreach \name/\pos/\value in {{E/(-3,-1)/{E(1)}},{B/(-1,-1)/{B(2)}},{D/(1,-1)/{D(3)}},{K/(3,-1)/{K(4)}},{F/(-1,-2.5)/{F(5)}}}
	\node[terminal,label=right:$\value$] (\name) at \pos {}  ;

\foreach \name/\pos in {{A/(0,0.5)}}
	\node[nonterminal,label=right:$\name$] (\name) at \pos {}  ;
 \coordinate [below right=0.4cm of B,label={right:${h_B[A]=3}$}] (y);

\foreach \source/\dest/\text in {A/E/{($A$,3)},A/B/{($A$,3)},A/D/{($A$,3)},A/K/{($A$,3)}}                                   
	\draw[->]
	  (\source) to node[above, sloped,font=\sffamily\tiny] {\text}  (\dest);	 		
 
 \foreach \source/\dest/\text  in {B/F/{($A$,3)}}                                   
   \path[->] (\source) edge node[below, sloped,font=\sffamily\tiny] {$\text$} (\dest) ; 
\end{tikzpicture}
\label{fig:protocol_sf4}
}
\subfloat[Broadcasting phase (cont.)]{\begin{tikzpicture} 
[>=stealth', thick, scale=0.7,  
nonterminal/.style={
	circle,
	minimum size=3mm,
	very thick,
	draw,
	font=\itshape
},
terminal/.style={
circle,
minimum size=3mm,
very thick,draw,
font=\sffamily}, 
every label/.style={thick}
]

\foreach \name/\pos/\value in {{E/(-3,-1)/{E(1)}},{B/(-1,-1)/{B(2)}},{D/(1,-1)/{D(3)}},{K/(3,-1)/{K(4)}}, {F/(-1,-2.5)/{F(5)}},{M/(-1,-4)/{M(6)}}}
	\node[terminal,label=right:$\value$] (\name) at \pos {}  ;

\foreach \name/\pos in {{A/(0,0.5)}}
	\node[nonterminal,label=right:$\name$] (\name) at \pos {}  ;
 \coordinate [below right=0.3cm of F,label=right:${h_F[A]=3,\mathcal{C}_F[A]=\{3,3,3\}}$] (y);
 
\foreach \source/\dest/\text in {A/E/{($A$,3)},A/B/{($A$,3)},A/D/{($A$,3)},A/K/{($A$,3)},E/F/{($A$,3)},B/F/{($A$,3)},D/F/{($A$,3)}}                                   
	\draw[->]
	  (\source) to node[above, sloped,font=\sffamily\tiny] {\text}  (\dest);	 		
 
 \foreach \source/\dest/\text  in {F/M/{($A$,3)}}                                   
   \path[->] (\source) edge node[below, sloped,font=\sffamily\tiny] {$\text$} (\dest) ; 
\end{tikzpicture}
\label{fig:protocol_sf5}
}
\hspace*{\fill}
\caption{{Polling algorithm for $i=k=1$ and $m=3$.}}
\label{protocol_figure}
 \vspace{-7mm}
\end{figure*}  

\noindent\textbf{Broadcasting.}   In this phase, each node needs to disseminate its collected data to all other nodes in such a way that each other node eventually obtains that correct data. 
In the naive approach, upon receiving a message from the neighbor, a node stores the data then forwards it on every other edge.
Despite the use of richer social graph structures, and with the presence of dishonest nodes which can corrupt data, one node can receive/send so many duplicated messages (which may be passed by many paths) from/to other nodes. This leads to flooding the local storage.
As motivated in Section \ref{subsec:graph_model}, we propose a method for efficiently broadcasting messages by using our concept of the $m$-broadcasting property. For a graph satisfying the $m$-broadcasting property, each node $n$ first
 sends its collected data to all neighbors (lines \ref{mark1st_br_proc}--\ref{mark1end_br_proc}).
Then, upon receipt of the message  containing the collected data of source $s$ from neighbor $r$ preceding it in the ordering (w.r.t. source $s$), node $n$ does one of the following activities:
\vspace{-2mm}
\begin{itemize}
\item $r=s$: It decides on the data of source $s$ by storing the value $c_s$ in  $h_n[s]$. When the value $h_n[s]$ is assigned, it is further forwarded  to all $d_n-\beta_n(s)$ nodes succeeding it in the ordering (lines \ref{mark2st_br_proc}--\ref{mark2end_br_proc}).  

In Fig. \ref{fig:protocol_sf4}, node $A$ broadcasts its data, then  $B$ receives, stores that data  in $h_B[A]$, and forwards it to $F$.

\item $r\ne s$: To avoid the case that value $h_n[s]$ might be calculated (and broadcast) twice for the direct neighbor $s$, node $n$ only considers the case  $r\ne s \wedge s\notin\Gamma(n)$. If this condition is satisfied, it adds the value $c_s$ to the multiset $\mathcal{C}_n[s]$ of possible collected data for  $s$ (line \ref{mark3st_br_proc}). When node $n$ has received the expected number $m$ of possible collected data for a given source $s$, it decides on the collected data by choosing the most represented value in $\mathcal{C}_n[s]$ and puts it in $h_n[s]$. (Since the decision is based on the most represented value, instead of waiting for receiving all $m$ forwarded data, node  $n$ can decide the source's data upon receipt of more than $m/2$ identical data.)  Node $n$ then further forwards the data to all $d_n-\beta_n(s)$ nodes succeeding it (lines \ref{mark4st_br_proc}--\ref{mark4end_br_proc}). 

Fig. \ref{fig:protocol_sf5} depicts node $F$ receives messages from its neighbors about the data of source $A$. It has four friends, but  receives only $m=3$ messages from preceding neighbors $E,B,D$.  As all values in $\mathcal{C}_F[A]$ are 3, node $F$ decides that value as the collected data of $A$ and stores it in $h_F[A]$. It then forwards that data to its succeeding node  $M$. 
\vspace{-1mm}
\end{itemize}
When a node decides the collected data of $s$ and has no succeeding friend, the value $h_n[s]$ is no longer forwarded.
This phase is complete if a node decides on the collected data of all other ones in the network.

\noindent\textbf{Aggregating.}
The final result is obtained at each node by simply 
doing this computation:
$\result=c_n+\sum_{j\ne n}h_n[j]$ (lines \ref{mark1st_aggr_proc}--\ref{mark1end_aggr_proc}).
\vspace{-3mm}

\section{Correctness and complexity analysis}\label{sec:correctness}
In this section, we   present the correctness and complexity analysis of our protocol when deployed on the graphs 
 without and with the presence of dishonest nodes.

\subsection{Protocol and graph without dishonest nodes}\label{sec:protocol_without_dishonest} 
We first analyze 
the correctness (including accuracy and termination) of our protocol for the graphs of $\mathcal{G}_1$ in which all participants are honest in Theorem \ref{lem:accuracy_idealcase}. Then we analyze the spatial and message
complexities in  Propositions 1 and 2 respectively. 


\begin{theorem}[Correctness] \label{lem:accuracy_idealcase}
Consider a polling system of size $N$ with only honest nodes where each node $n$  expresses a vote $v_n$.  The polling algorithm is guaranteed to eventually terminate and each node  outputs $\sum_{n=0}^{N-1} v_n$.
\vspace{-2mm}
\end{theorem}

\begin{proof}[Proof (Accuracy)]
In the sharing phase, each node $n$ sends a set of shares
$\mathcal{P}_n=\{p_{n_1},p_{n_2},...,p_{n_{2i+1}}\}$ to its  consumers where
 $\sum p_{n_j}=(i+1)\ldotp v_n+i\ldotp(-v_n)=v_n$, and also receives a set of shares
$\{p'_{n_1},p'_{n_2},...,p'_{n_l}\}$ from its producers ($l=|\mathcal{R}_n|$) to obtain a collected data of value $c_n = \sum_{j=1}^{l}p'_{n_j}$.
With the assumption there is no dishonest node and without crash or
message
loss, each message from the source successfully reaches the destination, and
thus the set of all sending shares of all nodes will be
exactly coincided with the set of all receiving shares of all nodes, namely:
$\bigcup_{V} \{p_{n_1},p_{n_2},...,p_{n_{2i+1}}\}
=\bigcup_{V}\{p'_{n_1},p'_{n_2},...,p'_{n_{l}}\} $.

In the broadcasting phase, each node $n$ broadcasts its collected data to their neighbors, then they do honestly forward that value to neighbors of neighbors of $n$ (succeeding it in the ordering w.r.t. source $n$) and so on. 
The messages   are finally received by all direct and indirect neighbors. It infers an array $h_n$ contains all
collected data of all nodes in the system. Consequently, the final computation gives us
the value:
$\result=c_n+\sum_{\substack{0\le j<N, j\neq n}}h_n[j]
=\sum_{j=0}^{N-1}{c_j}  
=\sum_{j=0}^{N-1}{\sum_{t=1}^{l}p'_{j_t}  }
=\sum_{j=0}^{N-1}{\sum_{t=1}^{2i+1}p_{j_t}}
=\sum_{j=0}^{N-1} v_j.
$

\noindent\emph{Proof (Termination).}
In the sharing phase, each node has to receive a finite number ($|\mathcal{R}|$)
of messages. 
In the broadcasting phase, for a source $s$, each direct neighbor of $s$ receives only one message, and each other indirect neighbor receives $m$ messages. 
Since every node sends the required number of messages and they are eventually arrived to destination, each phase completes. 
As the protocol has a limited number of phases, it is ensured to
eventually terminate.
\end{proof}


\begin{prop}[Spatial complexity]\label{propos:spatial_ideal}
The total space each node must hold is  $\mathcal{O}(k+m.N)$. 
\vspace{-2mm}
\end{prop}

\begin{proof}
Each node $n$  maintains a set of producers and consumers (at most $2(2k+1)$), a list of $d_n$ direct neighbors, a set of $N$ identities of nodes in the systems, a set of $m$  possible collected data for each of source $s$,  a set $h_n$ to store the deciding collected data. 
Therefore, the spatial complexity is $\mathcal{O}(k)+\mathcal{O}(d_n)+\mathcal{O}(N)+\mathcal{O}(m.(N-1))+\mathcal{O}(N-1)=\mathcal{O}(k+m.N)$.
\end{proof}

\begin{prop}[Message complexity]\label{propos:msg_ideal}
The  number of messages sent by a node $n$ is  $\mathcal{O}(k+N.(d_n-m))$.
\vspace{-2mm}
\end{prop}

\begin{proof}
In the sharing phase, node $n$ sends at most $2k+1$ messages. In the broadcasting phase, it sends $d_n$ messages containing its collected data, and forwards at most $d_n-m$ messages when receiving collected data of each source $s$ from its neighbors.    Thus, the  message complexity is
$\mathcal{O}(k)+\mathcal{O}(d_n)+\mathcal{O}((N-1).(d_n-m))=\mathcal{O}(k+N.(d_n-m))$.
\end{proof}
\vspace{-3mm}

\subsection{Protocol and graph with dishonest nodes}\label{sec:protocol_with_dishonest}

In this section, we study the impact of dishonest nodes on privacy and accuracy when  
EPol is deployed on  graphs of $\mathcal{G}_2$.  

\noindent\textbf{Privacy analysis}

\noindent We denote by $\gamma_i$ the proportion of nodes voting with $2i+1$ shares in the sharing phase, where $0\le i\le k$ and $\sum_{i=0}^{k}\gamma_i=1$.
Assuming $D\le (m-1)\Delta_G/2$ (presented later in Theorem \ref{theo:max_tolerance}),                                
we consider two cases for disclosing a node's vote as follows.

\noindent\textbf{Vote disclosing with certainty.}
We discuss the case when the vote of a given node may be disclosed with certainty in  the following theorem.


%

\begin{theorem}[Certain Privacy]\label{theo:privacy_generalcase}
Assume a coalition of $D$ dishonest nodes knows the number  of shares sent by a  node.
The probability $P_{ce}$  this coalition reveals correctly with certainty a vote of the honest node voting with $2i+1$ shares  ($0\le i\le k$)
is at most $\gamma_i  \bigl(\frac{D}{N}\bigr)^{i+1} $.
\vspace{-3mm}
\end{theorem}

\begin{proof}
The coalition reveals correctly a node's vote $v$ if and only if $i+1$ consumers receiving $i+1$ identical shares of value $v$ belong to the dishonest coalition. There are a proportion $\gamma_i$ of these nodes. Thus, that event occurs with probability $P_{ce}=\gamma_i {\binom{D}{i+1}}/{\binom{N}{i+1}}\le \gamma_i  \bigl(\frac{D}{N}\bigr)^{i+1}$.
\end{proof}

\begin{corollary} \label{cor:prob_reveal_certainty}
If all nodes send $2k+1$ shares in the sharing phase, then the probability that a coalition of $D$ dishonest nodes reveals correctly with certainty an honest node's vote 
is at most $ \bigl(\frac{D}{N}\bigr)^{k+1}$.
\vspace{-2mm}
\end{corollary}

\begin{proof}
The claim is followed by   Theorem \ref{theo:privacy_generalcase}.
\vspace{-1mm}
\end{proof}

We plot the bound of $P_{ce}$ as a function $f(\gamma_i,i)=\bigl(\frac{D}{N}\bigr)^{i+1}$ in Fig. \ref{fig:upperbound_privacy} for $k=3$, $N=100$ and $D=20$. We see that  $P_{ce}$ increases  with the increase of  $\gamma_i$ and the decrease of $i$. 
Thus, we  get the maximum privacy   when all nodes generate $2k+1$ shares, and the minimum privacy when all nodes generate only one share.
\vspace{-1mm}

\begin{figure*}[t]
\vspace{-8mm}
 \centering
 	\subfloat[Certainty]{\includegraphics[scale=0.25,angle=-90]{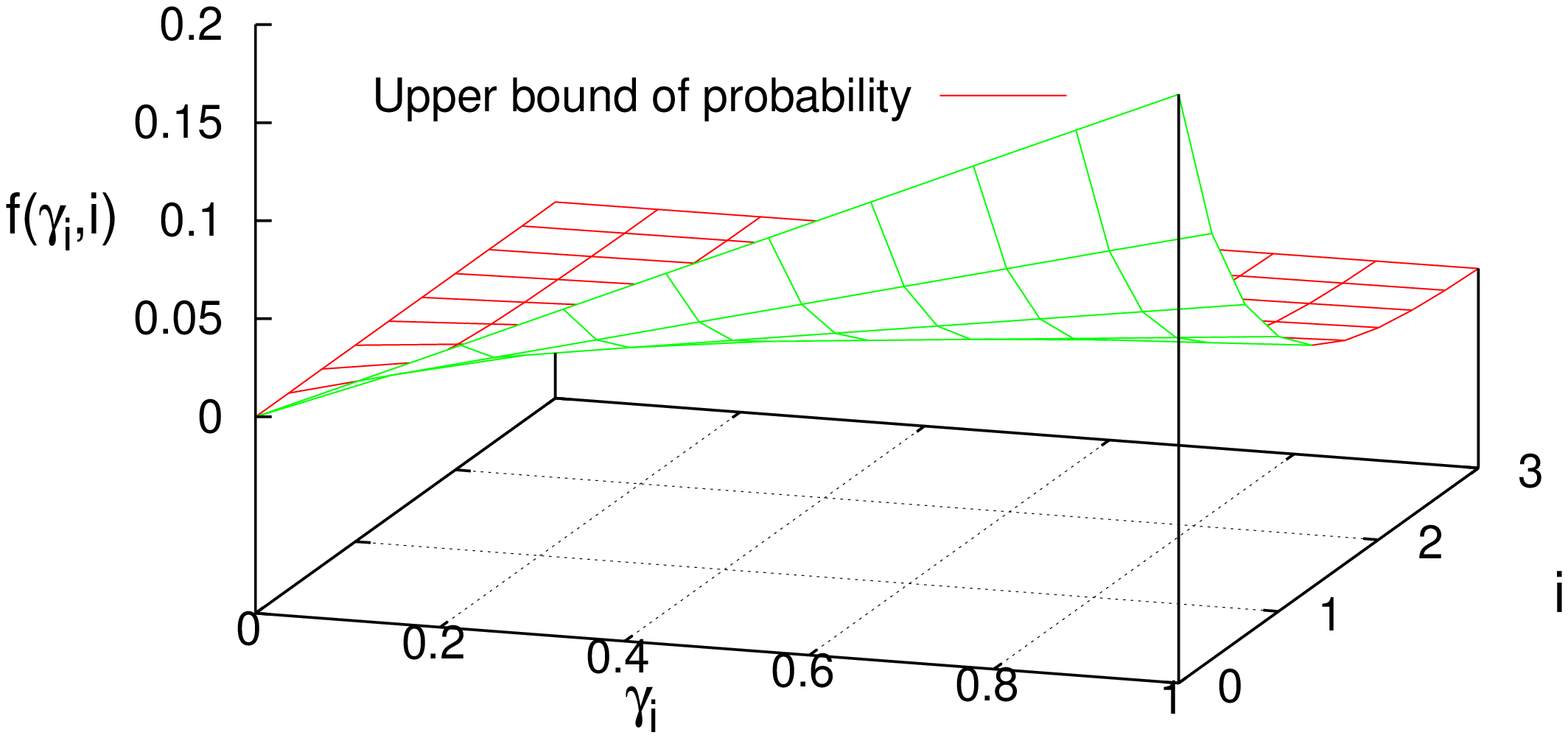}
 \label{fig:upperbound_privacy}}
 	\subfloat[Greedy]{\includegraphics[scale=0.25,angle=-90]{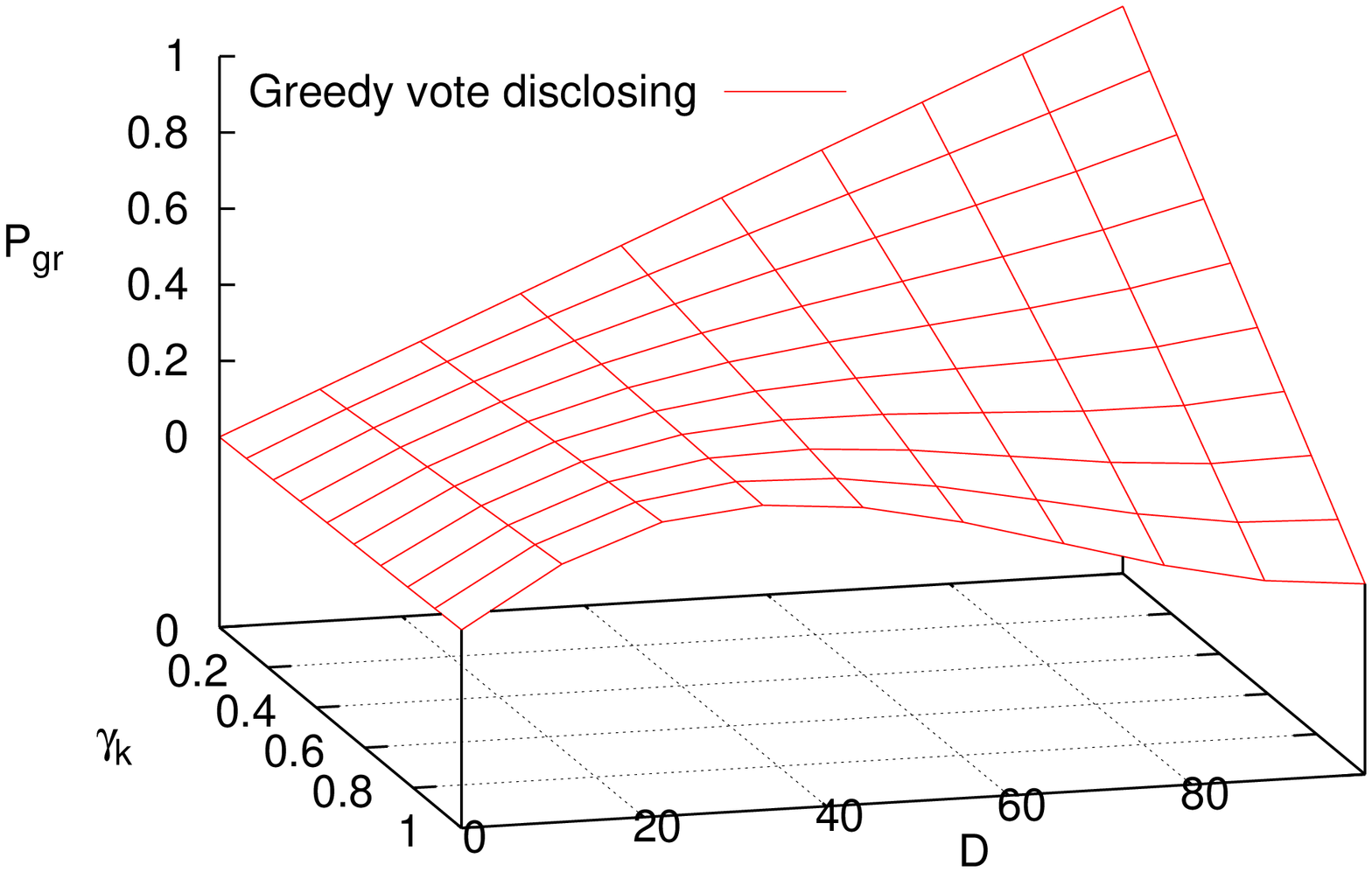}
  \label{fig:greedy_reveal_vary_gamma_D}}
 	\subfloat[Non-greedy]{\includegraphics[scale=0.25,angle=-90]{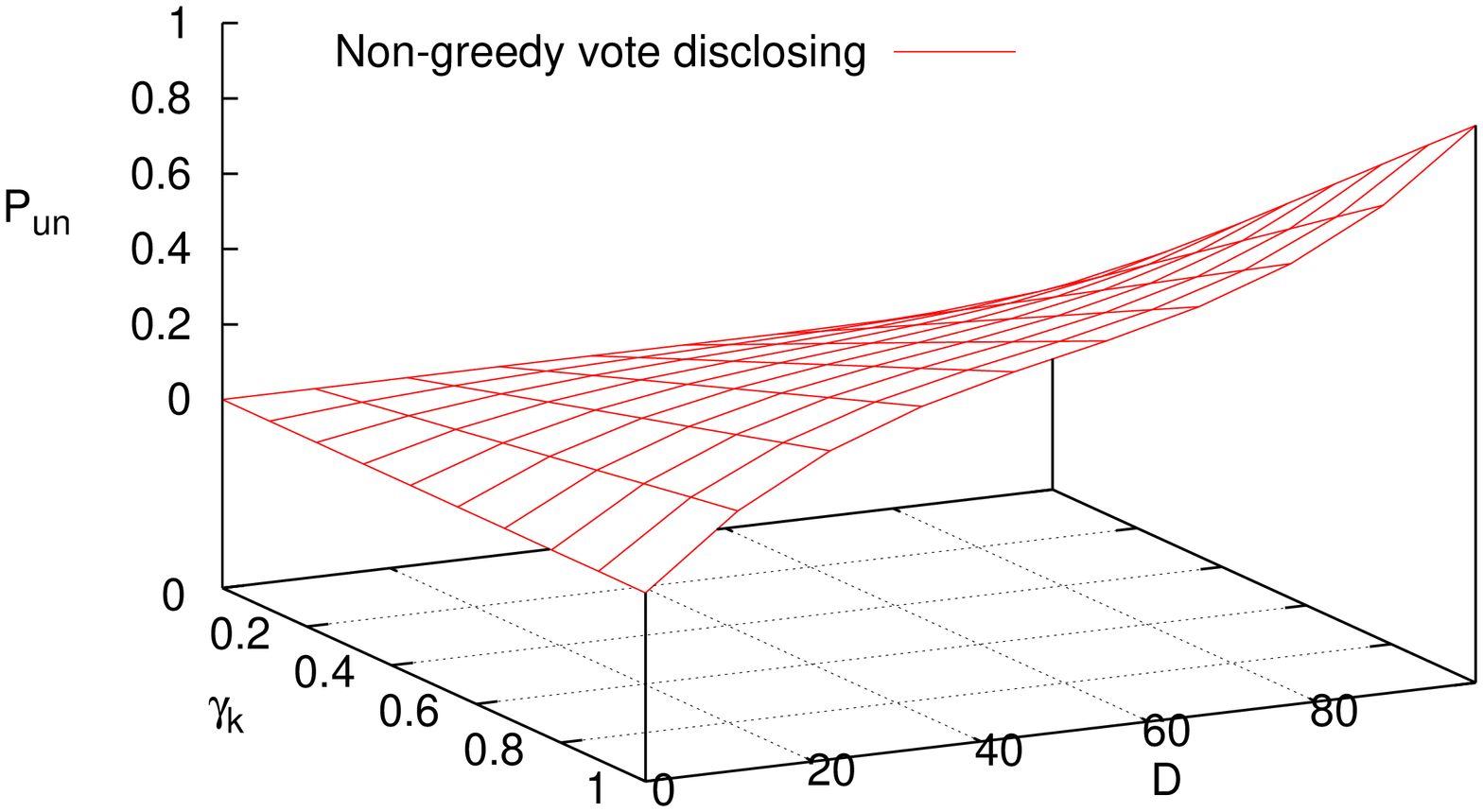}
  \label{fig:nongreedy_reveal_vary_gamma_D}}
  \caption{Probability to disclose a node vote.}
 \label{fig:privacy}
\vspace{-5mm}
 \end{figure*}


If the poll outcome is $N$ (resp. $-N$), it infers all nodes vote for ``$+1$'' (resp. ``$-1$'') and they all are disclosed. 
Moreover, w.l.o.g, assume dishonest nodes always vote for ``$-1$'', thus, if the result is $N-2D$ (resp. $-N$) then it implies all honest nodes vote for ``$+1$'' (resp. ``$-1$'').
Without considering these cases, i.e., all honest nodes do not vote for the same option, Theorem \ref{max_num_revealed_votes} shows us the maximum number of votes the dishonest coalition could discover.
\vspace{-1mm}

\begin{theorem}  \label{max_num_revealed_votes}
If all honest nodes do not vote for the same option, a coalition of $D$ dishonest nodes can reveal at most $2D$ votes of honest nodes.
\vspace{-2mm}
\end{theorem}

\begin{proof}
A consumer receives on average $\sum_{i=0}^{k} \gamma_i(2i+1)$ shares, hence, the dishonest coalition collects at most $D\sum_{i=0}^{k} \gamma_i(2i+1)$ shares. Moreover, a vote $v$ of one node voting with $2i+1$ shares is revealed if and only if the coalition obtains $i+1$ identical shares of value $v$. Thus it recovers at most $
\lfloor D.\frac{\sum_{i=0}^{k}\gamma_i(2i+1)}{\sum_{i=0}^{k}\gamma_i(i+1)} \rfloor\le 2D$ votes.
\end{proof}

\noindent\textbf{Vote disclosing without certainty.}
This part examines the case that the dishonest nodes collude to reveal an honest node's vote without sureness. The coalition decides a node's vote based on the received shares in the sense they can decide the vote after getting some shares or after getting all shares from that node. Thus, they choose one of the following two strategies: (a) Upon receipt of $\rho+1$ identical shares (for some $0\le \rho\le k$) from a given node, they will be considered as its vote; (b) After receiving all shares from a given node, the most represented value of the received shares will be considered as its vote. The former strategy is used by the ``greedy'' dishonest users who want to reveal rapidly the honest user's vote (even if they have just received one share). The latter one is used by the ``non-greedy'' dishonest users who are patient and wait for receiving all node's shares before trying to disclose the vote.
We present the probabilities that a coalition of dishonest nodes discloses an honest node's vote for these situations in Theorems \ref{theo:privacy_greedy} and \ref{theo:privacy_ungreedy}. Recall that each node does not know the number of shares generated by other ones and $\Delta_G$ is a network diameter.

\begin{theorem}[Greedy vote disclosure]\label{theo:privacy_greedy}
Assume a coalition of $D$ dishonest nodes  agrees on the following rule ``upon receipt of $\rho+1$ identical shares ($0\le \rho\le k$) from a given node, they will be considered as the node's vote''.   The probability this coalition reveals correctly a node's vote is 
$P_{gr}(\rho)=\sum_{i=\rho}^{k} \gamma_i\cdot{\binom{D}{\rho+1}}\sum_{j=0}^{\rho}
{\binom{D-\rho-1}{j}}/{\binom{N}{j+\rho+1}}
$
and is bounded by $\sum_{i=\rho}^{k}\gamma_i\frac{N+1}{N-D+\rho+2}  (\frac{D}{N-D+\rho+1}  ) ^{\rho+1}$.
\end{theorem}
\vspace{-2mm}

\begin{proof}
The dishonest nodes succeed to discover a node's vote $v$ if  that node has sent $2i+1\ge 2\rho+1$ shares in which $\rho+1$ identical ones representing $v$ and up to $\rho$ shares of value $-v$ were received by the dishonest consumers. 
In contrast, the coalition's decision is failed if the node has sent more than $2\rho +1$ messages (i.e., at least $2\rho+3$ messages)  but the dishonest nodes  obtained only $\rho+1$ messages of value $-v$ and up to $\rho$ messages of value $v$. 
The probability a coalition of $D$ dishonest nodes discloses correctly a vote $v$ is:
$P_{gr}(\rho)=\sum_{i=\rho}^{k}\gamma_i p(i)$ where
$p(i)=
{\binom{D}{\rho+1}}\sum_{j=0}^{\rho}
{\binom{D-\rho-1}{j}}/{\binom{N}{j+\rho+1}}
$.
Using this identity:
$1/\binom{n}{j}= (n+1) \int_0^1 t^{j}(1-t)^{n-j} dt$ for some positive $n$ and $j$, then $p(i)$ is rewritten as follows:
$p(i)=(N+1) \binom{D}{\rho+1} \sum_{j=0}^\rho \binom{D-\rho-1}{j} \int_0^1 t^{j+\rho+1}(1-t)^{N-j-\rho-1} dt
\leq \\ (N+1) \binom{D}{\rho+1} \int_0^1 (1-t)^{N-\rho-1} t^{\rho+1} \left[\sum_{j=0}^{D-\rho-1} \binom{D-\rho-1}{j} \left( \frac{t}{1-t}\right)^j \right] dt 
=   \frac{N+1}{N-D+\rho+2}      \binom{D}{\rho+1} /{\binom{N-D+\rho+1}{\rho+1}}
\le \frac{N+1}{N-D+\rho+2} \left(\frac{D}{N-D+\rho+1} \right) ^{\rho+1}.
$

This leads the desired result.
\end{proof}

In Theorem \ref{theo:privacy_greedy}, a vote $v$ of the honest node is discovered  if  that node has sent $2i+1\ge 2\rho+1$ shares in which $\rho+1$ identical ones representing $v$ and up to $\rho$ shares of value $-v$ were received by the dishonest consumers. Moreover, by Theorem \ref{theo:privacy_greedy}, value $P_{gr}$  increases when $\gamma_i$ decreases (and $i$ increases) and $D$ increases. 
For example, with $N=100$, $k=1$ (i.e., each node votes with one share or $2k+1=3$ shares), $\rho=0$, we plot the probability $P_{gr}$ as a function of $D$ and $\gamma_k$ in Fig. \ref{fig:greedy_reveal_vary_gamma_D}.
As expected, the vote privacy decreases (i.e.,  $P_{gr}$ increases) when 
 $\gamma_k$ decreases and $D$ increases. 

\begin{theorem}[Non-greedy vote disclosure]\label{theo:privacy_ungreedy}
Assume a coalition of $D$ dishonest nodes agrees on the following rule ``the most represented value of the received shares from a given node will be considered as the node's vote''.   The probability   this coalition reveals correctly a node's vote 
is $P_{un}=\sum_{i=0}^{k}\gamma_i\cdot\sum_{j=1}^{i+1} \sum_{t=0}^{j-1}  
 {\binom{D}{j}}{\binom{D-j}{t}} /{\binom{N}{j+t}}
 $ and is bounded by  $\bigl( \frac{D}{N}/(1-2\frac{D}{N})\bigr)\bigl[{1-\sum_{i=0}^{k}\gamma_i(2\frac{D}{N})^{2i+1}}\bigr]$.
\vspace{-1mm}
\end{theorem}

\begin{proof}
The dishonest nodes reveal successfully a   vote $v$ of a node voting with $2i+1$ shares if they receive $j$ shares of value $v$ and $t$ shares of value $-v$ such that $i+1\ge j>t\ge 0$. This event occurs with probability 
$p(i)=\sum_{j=1}^{i+1} \sum_{t=0}^{j-1}  {\binom{D}{j}}{\binom{D-j}{t}}/{N\choose j+t}$. 
We have: ${\binom{D}{j}}{\binom{D-j}{t}}/{N\choose j+t}\le \left(\frac{D}{N}\right)^{j+t}\binom{j+t}{j}$. 
We denote $a=D/N$ and find the upper bound of 
$\sum_{j=1}^{i+1}\sum_{t=0}^{j-1} a^{t+j}\binom{j+t}{j}$. 
Rewrite that expression as 
$\sum_{r=1}^{2i+1} a^r\sum_{i<j\wedge i+j=r} \binom{r}{i}\leq \sum_{r=1}^{2i+1}a^r \frac{1}{2}\sum_{i=0}^{r} \binom{r}{i} = \frac{1}{2}\sum_{r=1}^{2i+1} (2a)^r =  \frac{a}{1-2a}\bigl(1-(2a)^{2i+1}\bigr)$.
Since a node sends $2i+1$ shares with probability $\gamma_i$, and we consider all possibilities of value $i$, this  gives us the desired result $P_{un}=\sum_{i=0}^{k}\gamma_i p(i)\le    \frac{a}{1-2a}[1-\sum_{i=0}^{k}\gamma_i (2a)^{2i+1}]$.
\end{proof}

By Theorem \ref{theo:privacy_ungreedy}, the quantity $P_{un}$ increases when both values $\gamma_i$ and $D$ increase (and also $i$ increases). 
For $N=100$ and $k=1$,
Fig. \ref{fig:nongreedy_reveal_vary_gamma_D} shows the impact of $D$ and  $\gamma_k$  on $P_{un}$.  According to this result, as expected, the vote privacy decreases (i.e., $P_{un}$ increases)  when both 
$D$ and $\gamma_k$ increase.

\noindent\textbf{Combining vote disclosure with and without certainty.}
The 
dishonest nodes may try to reveal a node's vote either in certainty or uncertainty.  Assume they respect the vote disclosure rule with and without certainty. From the viewpoint of dishonest nodes, they always want their vote detection to be as certain as possible, i.e., they prefer a node's vote being revealed with certainty to other cases. Hence
their strategy is as follows: they first try to disclose a vote of node with sureness. If they do not succeed, for instance, because of lacking of messages, they will consider the way to detect that vote without certainty.
It implies the probability  a vote is disclosed in this case is $P_{com}=\max\{P_{ce},P_{gr},P_{un}\}$.


\noindent\textbf{Accuracy.}
In this section, we evaluate the maximum and average impact on accuracy caused by the dishonest coalition when we deploy EPol on the graphs of $\mathcal{G}_2$. 
A dishonest node may affect the poll outcome with the following \emph{misbehaviors}:
\vspace{-5mm}
\begin{enumerate}[(i)]
 \item   Since a node can only generate and send shares to its consumers it is assigned (otherwise the shares are dropped) and there are at most $2k+1$ consumers to be assigned, it must send at most $2k+1$ shares in which at most $k+1$ ones are identical. Hence a dishonest node may give the misbehavior by  sending more than $k+1$ (but not greater than $2k+1$) identical shares.
\item   It inverts each receiving ``$+1$''-share into a ``$-1$''-share to decrease the collected data.
 \item   It modifies the collected data of other honest node
or sends a forged message in the broadcasting phase.
\end{enumerate}
\vspace{-1mm}
\noindent \emph{Verification procedures}: in the first attack, the worst case is when the dishonest node sends all $2k+1$ shares of value ``$-1$''. In the attack (ii),  a node receives at most $2k+1$ shares (since at most $2k+1$ producers are assigned) and thus, the computing collected data must be inside the range $[-(2k+1),2k+1]$. The misbehaviors (i) and (ii) cannot be detected  without inspecting the content of the shares themselves, but the misbehavior (iii) is detected with certainty if the dishonest nodes transmit or corrupt the collected data outside the range $[-(2k+1),2k+1]$. 
Noted that we do not consider the Sybil attacks, 
hence, in case (iii) a dishonest node cannot create a forged message containing identity of other node. Moreover, this activity does not affect the final result since a node receives directly a  message from source $s$ (if it is a neighbor of $s$) or  gets a majority of receiving messages ($\lceil (m+1)/2 \rceil$) containing the correct data of $s$.  
We show the impact of these misbehaviors on accuracy in Theorems \ref{total_impact_theorem} and \ref{theo:avg_impact}.
\vspace{-1mm}

\begin{theorem}[Maximum impact]\label{total_impact_theorem}
Each dishonest node may affect the final result up to $6k+4$. 
\vspace{-3mm}
\end{theorem}

\begin{proof}
In cases (i) and (ii) of the dishonest misbehaviors, the worst case of these misbehaviors occurs when the dishonest node sends $2k+1$  shares of value ``$-1$'' and inverts all the receiving ``$+1$''-shares into ``$-1$''-shares before computing the collected data. As a node is allowed to send only a set of shares that can regenerate to the vote of value ``$-1$'' or ``$+1$'', the first and the second attacks respectively affect up to $2k+1-(-1)=2k+2$ (if a node votes for ``$+1$'' but it generates $2k+1$ shares of value ``$-1$'') and $(2k+1)-(-2k-1)=4k+2$ (if it gets $2k+1$ shares of value ``$+1$'' but inverts all into the shares of value ``$-1$''). 
In the misbehavior (iii), a dishonest node only modifies the collected data of other honest node $s$ such that the data is in the range $[-2k-1,2k+1]$, otherwise its misbehavior will be detected. However this activity does not affect the final result since one node receives a direct message from $s$ (if it is a neighbor of $s$) or receives $m$ forwarding messages from neighbors in which at most $\lfloor (m-1)/2\rfloor$  messages are corrupted, and majority messages ($\lceil (m+1)/2 \rceil$) contain correct data, thus an honest node always obtains the correct collected data of $s$.
Therefore, the maximum impact from one dishonest node is $2k+2+4k+2=6k+4$.
\end{proof}


\begin{theorem}[Average impact]\label{theo:avg_impact}
Let  $\alpha$   be the proportion of nodes voting for ``$+1$''. The average impact from a dishonest node is $I_{avg}=\bigl[\sum_{i=0}^{k}\gamma_i (2i+1)\bigr]\bigl[1+2\sum_{i=0}^{k}\gamma_i\frac{i+\alpha}{2i+1}\bigr]+1$.
\vspace{-3mm}
\end{theorem}

\begin{proof}
Each node generates $i+1$ shares of its tendency, and $i$ opposite shares, thus, the consumers receive a ``$+1$''-share with probability $\sum_{i=0}^{k} \gamma_i[\alpha\frac{i+1}{2i+1}+(1-\alpha)\frac{i}{2i+1}]=\sum_{i=0}^{k} \gamma_i\frac{i+\alpha}{2i+1}$. In addition, as the average number of receiving shares is $\sum_{i=0}^{k}\gamma_i (2i+1)$, the average number of ``$+1$''-shares is $\bigl[\sum_{i=0}^{k}\gamma_i (2i+1)\bigr]\bigl[\sum_{i=0}^{k}\gamma_i\frac{i+\alpha}{2i+1}\bigr]$. By altering any receiving ``$+1$''-share into ``$-1$''-share this impacts  $(1-(-1))\bigl[\sum_{i=0}^{k}\gamma_i (2i+1)\bigr]\bigl[\sum_{i=0}^{k}\gamma_i\frac{i+\alpha}{2i+1}\bigr]=2\bigl[\sum_{i=0}^{k}\gamma_i (2i+1)\bigr]\bigl[\sum_{i=0}^{k}\gamma_i\frac{i+\alpha}{2i+1}\bigr]$. Moreover, a dishonest node sends on average $\sum_{i=0}^{k}\gamma_i (2i+1)$ shares of value ``$-1$'' which affects by $1+\sum_{i=0}^{k}\gamma_i (2i+1)$. Consequently, the total average  impact from a dishonest node is $I_{avg}=\bigl[\sum_{i=0}^{k}\gamma_i (2i+1)\bigr]\bigl[1+2\sum_{i=0}^{k}\gamma_i\frac{i+\alpha}{2i+1}\bigr]+1$.
\end{proof}

The quantity $I_{avg}$ is minimized when all nodes generate the same number of shares, e.g., $2i+1$, and thus $I_{avg}=2(2i+\alpha+1)$.
In the worst case, a dishonest node sends $2k+1$ shares, the minimized average impact is $I_{avg}=2(2k+\alpha+1)$.
 
\begin{corollary}\label{cor:bias_2k+1}
The expected biased result is in the range $[(2\alpha-1)N-(6k+4)D;$ $(2\alpha-1)N]$. More particularly, if all nodes send $2k+1$ shares, then the expected biased result is $(2\alpha-1)N-(4k+2\alpha+2)D$.
\vspace{-3mm}
\end{corollary}

\begin{proof}
The expected outcome is   $\alpha N.1+(1-\alpha)N.(-1)=(2\alpha-1)N$. 
By Theorem \ref{total_impact_theorem}, the maximum impact of dishonest nodes is $(6k+4)D$, hence, the biased result must be in the range $[(2\alpha-1)N-(6k+4)D;$ $(2\alpha-1)N]$. In addition, by
Theorem \ref{theo:avg_impact}, if all nodes send $2k+1$ shares then total average impact is $(4k+2\alpha+2)D$, and this yields the proof.
\end{proof}

By Theorems \ref{theo:privacy_generalcase}--\ref{theo:avg_impact} and Corollaries \ref{cor:prob_reveal_certainty}--\ref{cor:bias_2k+1}, for a fixed parameter $k$,  the number of users
voting with a high number of shares (e.g., $2k+1$ shares) affects the privacy and accuracy.
More concretely, if we care about vote privacy, we should augment the number of nodes generating $2k+1$ shares since the probability to disclose a node's vote  with certainty ($P_{ce}$) and with greedy uncertainty ($P_{gr}$) will decrease. But this rises up the probability $P_{un}$ of non-greedy vote disclosure and the impact on the final outcome.
In contrast, if we take care of the accuracy of the final result, we should decrease the number of nodes voting with $2k+1$ shares since that reduces the impact on the final outcome. 
It also decreases $P_{un}$. 
However this increases the values $P_{ce}$ and $P_{gr}$. 


\noindent\textbf{Security.}
In this part,
we compute the maximum number of dishonest nodes that EPol can tolerate. 

\begin{lemma}\label{lem:bound_m-1}
A node decides correctly the collected data of an honest source $s$  if it connects directly to $s$ or there are at most $(m-1)/2$ dishonest neighbors preceding it (in the ordering w.r.t. source $s$).
\vspace{-3mm}
\end{lemma}
 
 \begin{proof}
 It is a trivial case when a node connects directly to source $s$. We only consider the case where a node $n\notin\Gamma(s)$. Since a node gets randomly $m$ messages among $\beta_n(s)$ ones, and its decision is based on the majority appearance of a value, it is necessary that less than $m/2$ messages are corrupted. This infers the desired result.
 \end{proof}
 	
\begin{theorem}[Tolerance to dishonest nodes]   \label{theo:max_tolerance}
The maximum number of dishonest nodes that EPol can tolerate is $\frac{m-1}{2}\Delta _G$.
\vspace{-3mm}
\end{theorem}

\begin{proof}
We denote $\Gamma_n^1=\Gamma(n)$, and define recursively the set of friends at distance $j$ of $n$, where $j>1$ by $\Gamma^j_n=\{u \mid \delta(u,n)=j\}=  \{u\mid u\in\Gamma_v \wedge v\in\Gamma_n^{j-1} \wedge u\notin \bigcup_{k<j} \Gamma_n^k \}$.
We define by $\varphi_n(s)$ the set of $n$'s neighbors preceding $n$ in the ordering w.r.t. source $s$. We have $|\varphi_n(s)|=\beta_n(s)$.
Let $dist(u,v)$ denote the distance, i.e., the length of the shortest path, between nodes $u$ and $v$.
By Lemma \ref{lem:bound_m-1},  node $n$  decides correctly about the collected data of honest source $s$  if $n\in\Gamma(s)$ or $|\varphi_n(s)\cap\mathcal{D}|\le (m-1)/2$.
To find out the highest number of dishonest nodes in the system we only consider the case $n\notin \Gamma(s)$. 
In general,  node $u\in\varphi_n(s)$  may be at any distance from $s$ (even $dist(u,s)\ge dist(n,s)$). It means $n$ may receive a corrupted data from a dishonest node at any distance from $s$. Since $n$ receives at most $(m-1)/2$ corrupted messages, the necessary condition to guarantee $n$ decides correctly $s$'s collected data is there are at most $(m-1)/2$ dishonest friends of $s$ at distance $i$, i.e., $|\Gamma_s^i\cap\mathcal{D}|\le (m-1)/2$ where $1\le i \le Rad(s,G)$ and $Rad(s,G)=\max_{u\in V}\{dist(s,u)\}$. 
This infers the maximum number of dishonest nodes such that each node receives correct collected data of source $s$ is $D(s)=\sum_{i=1}^{Rad(s,G)}|\Gamma_s^i\cap\mathcal{D}|\le \frac{m-1}{2}Rad(s,G)$. Considering all source nodes in the network, we have $D\le\max_{s\in V}\{D(s)\}=\frac{m-1}{2}\Delta _G$. 
\end{proof}
Notice that our protocol can tolerate more than $\sqrt{N}$ dishonest nodes for a ring-based structure introduced in \cite{DBLP:conf/opodis/GuerraouiHKM09,DBLP:journals/jpdc/GuerraouiHKMV12} (unlike these works tolerate less than $\sqrt{N}$ dishonest nodes). Indeed, as this structure has the diameter $\Delta_G=\sqrt{N}$, and with parameter $m\ge 3$, the upper bound of $D$ is not less than $\sqrt{N}$.

\begin{corollary} \label{cor:prob_converge}
If $D\le\frac{m-1}{2}\Delta_G$ then a node decides wrongly the collected data of some other node with the probability converging to 0 exponentially fast in $N$ (and $\Delta_G$).
\vspace{-3mm}
\end{corollary}

\begin{proof}
We define by $\varphi_n(s)$ the set of $n$'s neighbors preceding $n$ in the ordering w.r.t. source $s$.
By Lemma \ref{lem:bound_m-1}, if $|\varphi_n(s)\cap\mathcal{D}|\ge m/2$ then a node $n$ may decide wrongly the collected data of source $s$ (or could not make decision). 
We inspire the idea from \cite{DBLP:journals/jpdc/GuerraouiHKMV12} to compute the probability this event occurs by using standard Hoeffding bounds for sampling from a finite population without replacement (with the notice that $\Delta_G\le N-1$): 
$p_s = Pr\Bigl[|\varphi_n(s)\cap\mathcal{D}|\ge m/2\Bigl]
= Pr\Bigl[|\varphi_n(s)\cap\mathcal{D}|-\frac{D}{\Delta_G}\ge \frac{1}{\Delta_G}(\frac{m}{2}\Delta_G-D)\Bigr]
\le \exp(\frac{-2}{\Delta_G}(\frac{m}{2}\Delta_G-D)^2 )\le \exp(\frac{-2}{N-1}(\frac{m}{2}\Delta_G-D)^2) .
$
\noindent As $D<\frac{m}{2}\Delta_G$, the right-hand function tends to 0 exponentially fast in N (and $\Delta_G$).
It implies the probability for a node decides wrongly the collected data of some source (or could not make decision) is
$P_c=Pr\Bigl[ \bigcup_{s\in V}\{ |\varphi_n(s)\cap\mathcal{D}|\ge m/2\}\Bigr]
 \le \sum_{s\in V}p_s\le N\exp(\frac{-2}{N-1}(\frac{m}{2}\Delta_G-D)^2).
$
The right-hand function  tends to 0 exponentially fast in $N$.
\end{proof}

\noindent\textbf{EPol  vs. graphs.}\label{subsec:necessary_cond_ideal}
We examine the necessary and sufficient condition for our protocol to be deployed safely and securely in the following Theorem.

\begin{theorem} \label{theo:nec_suf_normal}
The properties of $\mathcal{G}_1$ and $\mathcal{G}_2$ are respectively the necessary and sufficient condition for
EPol to be deployed safely and correctly in the system without and with the presence of dishonest
nodes.
\vspace{-3mm}
\end{theorem}

\begin{proof}
We here just show this theorem with graphs of $\mathcal{G}_2$. For graphs of $\mathcal{G}_1$, the proof is easily considered as its consequences.
The sufficient condition $(\Leftarrow)$ is
proved by Theorems \ref{lem:accuracy_idealcase}--\ref{theo:max_tolerance}. 
We only need to
clarify the necessary condition $(\Rightarrow)$. Assume we have a general graph $G$.
We approach the proof by sketching step by step the requirements $G$
should obtain so that all properties of protocol are guaranteed.
In the sharing phase, node $n$ sends (or receives) at most $2k+1$ messages,
i.e.,
$|\mathcal{S}_n|=|\mathcal{R}_n|\le 2k+1$. In addition, it has to send an odd number of messages, thus, $|\mathcal{S}_n|=2i+1$ where $0\le i\le k$.
As $|\mathcal{S}_n|$ and $|\mathcal{R}_n|$ might not be disjoint, then $d_n\ge 2k+1$.
Therefore, to apply protocol correctly, $G$ must have the property $P_{g_1}$.
In addition, in the broadcasting phase of EPol, a node decides a data of some source $s$ if it connects directly to $s$, or to at least $m$ nodes which have preceding orders. This infers the graph must satisfy the $m$-broadcasting condition.
Finally, the Lemma \ref{lem:bound_m-1} shows  the necessary property $P_{g_3}$ that a graph must satisfy to deploy EPol.
Accordingly, the graph $G$ is a member of family $\mathcal{G}_2$ of graphs.
\end{proof}
\vspace{-2mm}

\section{Crash and message loss analysis} \label{sec:crash_msg_loss}
Nodes communicate by UDP which may suffer message loss on the communication channels. In addition, nodes may be unreliable, causing expected messages to be lost due to sender crashes.
Assuming the presence of node crashes and message losses in the system, this part analyzes the effect of these factors on  accuracy and termination of the protocol. Here we assume that the system contains no dishonest nodes.
    
\noindent\textbf{Impact on termination.}
Suppose a node is crashed with probability $r$ (and it is never recovered from a crash),  a message is lost (throughout transmitting)  with probability $l$, a node    \emph{fails to send shares} to its consumers with probability $q=r+(1-r)l$.
A node $n$ \emph{fails to decide} a collected data of source $s$ since:
\vspace{-3mm}
\begin{enumerate}
 \item $n=s$: $s$ fails to compute its collected data $c_s$. 
 \item $n\in\Gamma(s)$: $n$ does not receive a broadcast message from $s$.
 \item $n\notin\Gamma(s)$:  
more than $\beta_n(s)-m$ (preceding) neighbors \emph{fail to forward} the collected data as they either: (i) crashed, or (ii) have themselves not decided on the collected data, or (iii) have forwarded messages but they are lost.
\vspace{-2mm}
\end{enumerate}

We define by $e_{n_i}$ and $z_{n_i}$ respectively  the probability for a node $n$ at distance $i$ from source $s$ to fail to forward and fail to decide a  collected data of $s$.
We have $e_{n_i}=r+(1-r)[z_{n_i}+(1-z_{n_i})l]$,
where $z_{n_0}=z_s=\sum_{j=0}^{|\mathcal{R}_s|-1} \binom{|\mathcal{R}_s| }{j}(1-q)^j q^{|\mathcal{R}_s|-j}$, 
$z_{n_1}=e_{n_0}$,                                                                                                  
and        
$z_{n_i}=\sum_{j=0}^{m-1}\Bigl[\binom{\beta_n(s)}{j} \prod_{t=1}^{j} (1-e_{{n_t}_{l_t}}) \prod_{p=j+1}^{\beta_n(s)} e_{{n_p}_{l_p}}\Bigl]
$
where $i\ge 2$, $\{n_1,n_2,\dots,n_{\beta_n(s)}\}$ and $\{l_1,l_2,\dots,l_{\beta_n(s)}\}$ are respectively the sets of preceding friends of $n$ (w.r.t.  $s$) and their corresponding distances to $s$. Notice that $l_j$ could be greater than $i$ for $j=1,2,...\beta_n(s)$.
A node does not decide on the outcome if it has not decided on at least one collected data
of some source, that is $z_{n_{\Delta _G}}$.

\noindent\textbf{Impact on accuracy.}
A node crash affects  the final result when that node has some unique information and they  have not yet  been replicated, typically the shares of votes. It crashes (i) while sending its shares to consumers, or (ii) after summing up  shares from producers. The former case affects the final outcome up to $k+1$ (if it crashes after sending $k$ shares of $-v$). The latter case affects up to $2k+1$. Hence, the impact of such a crash on final result is at most $3k+2$.
\vspace{-3mm}


\section{Related work} \label{sec:related_work}
Several recent works related to  secret sharing and distributed polling have been proposed.
We introduce some works which are not based on any overlay structure and heavy computation. 
Secret sharing schemes \cite{DBLP:conf/crypto/Benaloh86a,DBLP:journals/cacm/Shamir79} may be used for polling with respect to addition. However, as they do not give the protection for the initial shares, the outcome is likely impacted with the presence of dishonest nodes. 
Verifiable Secret Sharing Scheme (VSS) 
and Multi-party Computation protocol (MPC) in \cite{DBLP:conf/focs/ChorGMA85,DBLP:conf/stoc/RabinB89}  
privately compute  the node's shares and give output with small error
if a majority of  nodes is honest. Nonetheless, without the condition of the initial input, a dishonest node can share arbitrary data, and bias the output. These protocols also use cryptography. This drawback is also applied for other studies based on MPC such as \cite{DBLP:conf/eurocrypt/CramerFSY96,DBLP:journals/ett/CramerGS97,DBLP:conf/crypto/DamgardIKNS08,DBLP:conf/crypto/DamgardN07} even  
the time and communication complexity are improved. Authors of  \cite{DBLP:conf/fc/MalkhiP01} proposed AMPC which provides users anonymity without using  cryptography, but this work used the notion of group. Based on AMPC and enhanced check vectors, E-voting protocol 
\cite{DBLP:conf/fc/MalkhiMP02} is the information-theoretically secure protocol. But it defines different roles for users and  thus, is different from our direction. The distributed ranking schemes are also related to our concern. 
However, they try to design accurate grading mechanism rather than providing efficient polling schemes \cite{Gupta:2003:RSP:776322.776346,DBLP:conf/colcom/Rodriguez-PerezEM05} and not address to privacy \cite{Dutta03thedesign}.  
In \cite{DBLP:conf/isit/ShahRR13}, the $m$-propagating condition enables the use of minimal shares for the secret. But, in our work, it is used to create a majority for deciding the  correct value during the broadcasting phase of our protocol.
The protocol in \cite{DBLP:conf/srds/GambsGHHK12} and \emph{AG-S3} \cite{DBLP:conf/journal/inforcom/Giurgiu20143} can be used for polling 
in a scalable and secure way, but they either use (i) a ring-based overlay, or (ii)  cryptography.  
\vspace{-3mm}

\section{Conclusion} \label{sec:conclusion}
This paper presented EPol, a distributed polling protocol for a general family of social graphs, while preserving vote privacy and limit the impact on accuracy of the polling outcome.
Unlike other
work, our protocol is deployed in   a more general family of graphs, and we obtained 
some similar and better results. 
In addition, we introduced some uncertain vote disclosure rules for dishonest nodes, and presented the probabilities of vote detection in these cases. We also analyzed the effect of message losses and nodes crashes on accuracy and termination.
Despite the use of richer social graph structures, the communication and spatial complexities of EPol are close to be linear.
In future work, we plan to implement our suggested protocol in some decentralized social networks like Diaspora and Tent.
\vspace{-3mm}


\bibliographystyle{IEEEtran}
\bibliography{authors_short}
\appendix

\subsection{Examples of the networks}
\label{appxsub:particular_graphs}

The algorithm EPol requires the graphs in a family satisfying the $m$-broadcasting property. Here we illustrate some particular graphs of this family. The examples of these graphs for $m=3$ are also depicted in Fig.\ref{fig:example_graphs}  and they can be generalized to any value $m$.
\vspace{-3mm}
\begin{enumerate}[(a)]
 \item 
\medskip\noindent\textbf{Layered networks.}
Each layer contains at least $m$ nodes, and each node links to all other nodes in neighboring layer. The number of nodes in each layer could be different. 
For each source node, an (increased) ordering satisfying the  $m$-broadcasting-source property w.r.t. this source is the ordering of the nodes with respect to the distance from the source.

\item
\medskip\noindent\textbf{Networks with a backbone.}
The graph includes a \emph{backbone} which is densely connected graph, and other nodes outside the backbone which connect directly to at least $m$ nodes in the backbone. 
An ordering satisfying the  $m$-broadcasting-source property w.r.t. one source is the ordering under the backbone subgraph, followed by all remaining nodes not in the backbone.

\item
\medskip\noindent\textbf{One-dimensional geometric networks.}
All nodes are arranged along a line, and there is a connection between two nodes if their distance is smaller than a fixed threshold. Each node has at least $m$ connections. 
An ordering of nodes w.r.t. one source is the order of their euclidean distance from the source.

\item
\medskip\noindent\textbf{Cluster-ring-based graph \cite{DBLP:conf/srds/BenkaouzGEH13,DBLP:conf/trustcom/EnglertG13,DBLP:conf/opodis/GuerraouiHKM09,DBLP:journals/jpdc/GuerraouiHKMV12}. }
The $N$ nodes are clustered into $r=\sqrt{N}$ \emph{ordered} groups, from $g_0$ to $g_{r-1}$. Each group is a clique. A node $n$ in group $g_i$ also links to a fixed-size set $\mathcal{S}_n$ of nodes in the next group ($\mathcal{S}_n\subset g_{i+1 \emph{ mod } r}$), and a fixed-size set $\mathcal{R}_n$ of nodes in the previous group where $|\mathcal{S}_n|=|\mathcal{R}_n|=2k+1$. Thus, all groups virtually form a ring with $g_0$ being the successor of $g_{r-1}$.
In fact, this structure is a particular case of layered networks presented above in which each layer is a clique of size $\sqrt{N}$, and all layers virtually form a ring.
Hence for each source node, an (increased) ordering satisfying the  $m$-broadcasting-source property w.r.t. a source is the ordering of the nodes with respect to the distance from the source. It is also ordered based on the distance from the source and the direction of sending messages, e.g., from group $g_i$ to $g_{i+1 \emph{ mod } r}$. 

\begin{figure}[ht]
\centering
\subfloat[Layered network]{
\begin{tikzpicture} 
[>=stealth',inner sep=0pt, scale=0.8,
dishonestnode/.style={
	circle, 
	minimum size=3mm,   
	thick,	
	draw, 
	fill,
	font=\sffamily
},
honestnode/.style={
circle, 
minimum size=3mm,
thick,
draw,
font=\sffamily
},every label/.style={thick}
]
\def \n {6}
\def \radius {1.5cm} 
\def \margin {45} 

\foreach \j in {0,1,2}
 \foreach \s in {0,1,2,3,4,5} 
{
  \node[honestnode] (v\j\s) at (\s,-\j) {};        
}

\coordinate (y1) at (-1,-1) {};
\coordinate (y2) at (6,-1) {};

\foreach \j/\s in {00/01,00/11,00/21, 10/01,10/11,10/21, 20/01,20/11,20/21,
01/02,01/12,01/22, 11/02,11/12,11/22, 21/02,21/12,21/22, 
02/03,02/13,02/23, 12/03,12/13,12/23, 22/03,22/13,22/23,
03/04,03/14,03/24, 13/04,13/14,13/24, 23/04,23/14,23/24,
04/05,04/15,04/25, 14/05,14/15,14/25, 24/05,24/15,24/25}
{
  \path (v\j) edge (v\s) ;       
} 

\path[-,dashed] (v10) edge (y1) ;   
\path[-,dashed] (v15) edge (y2) ; 
   
\end{tikzpicture}
\label{fig:example_graph_sf1}
}
\subfloat[Backbone network]{
\begin{tikzpicture} 
[>=stealth',inner sep=0pt,  
honestnode/.style={
circle, 
minimum size=3mm,
thick,
draw,
font=\sffamily
},every label/.style={thick}
]
\def \n {10}
\def \radius {1.5cm} 
\def \margin {45} 

\foreach \s in {1,3,4,6,8,9,10}{
  \node[honestnode] (v\s) at ({(360/\n * (\s - 1))}:\radius) {};

   \foreach \angle in {180,150,210}
      \draw[-] (v\s) -- ({(360/\n * (\s - 1))+\angle}:0.5);
}

\node[draw, cloud,  
   cloud puffs=10.9, cloud ignores aspect,minimum width=2.2cm,
    minimum height=1.75cm, fill=lightgray](cloud){};

\foreach \s in {1,...,5}
  \node[honestnode] (u\s) at ({(360/5 * (\s - 1))}:0.5) {};

\foreach \source/\dest in {u1/u2,u1/u3,u1/u4,u1/u5,u2/u3,u2/u4,u2/u5,u3/u4,u3/u5,u4/u5}
	\path[-] (\source) edge (\dest) ;

%

\end{tikzpicture}
\label{fig:example_graph_sf2}
}
\\
\subfloat[One-dimensional geometric network]{
\begin{tikzpicture} 
[>=stealth',inner sep=0pt, scale=0.8,
honestnode/.style={
circle, 
minimum size=3mm,
thick,
draw,
font=\sffamily
},every label/.style={thick}
]
\def \n {6}
\def \radius {1.5cm} 
\def \margin {45} 

 \foreach \s in {0,...,8} 
{
  \node[honestnode] (v\s) at (\s,0) {};        
}

\coordinate (y1) at (-1,0) {};
\coordinate (y2) at (9,0) {};

\foreach \j/\s in {0/1,1/2,2/3,3/4,4/5,5/6,6/7,7/8}
    \path (v\j) edge (v\s) ; 

\foreach \j/\s in {0/2,2/4,4/6,6/8}
    \path (v\j) edge [bend left=40] node[below, font=\sffamily\tiny] {} (v\s) ; 
\foreach \j/\s in {1/3,3/5,5/7}
    \path (v\j) edge [bend right=40] node[below, font=\sffamily\tiny] {} (v\s) ; 

\path[-,dashed,dotted] (v0) edge (y1) ;   
\path[-,dashed,dotted] (v8) edge (y2) ; 
   
\end{tikzpicture}
\label{fig:example_graph_sf3}
}
\hfill
\subfloat[Cluster-ring-based network]{
\begin{tikzpicture}
[>=stealth',inner sep=0pt,node distance=4mm, scale=0.8,
dishonestnode/.style={
	circle, 
	minimum size=3mm,   
	thick,	
	draw=black!90, 
	fill,
	font=\sffamily
},
honestnode/.style={
circle, 
minimum size=3mm, 
thick,
draw,
font=\small
},
every label/.style={thick}
]

\node[circle,minimum size=1.8cm,draw, dashed,label=below:$g_i$] (gi) at (0,0)  {};
\node[honestnode] (v11) at (-0.5,0.5) {}  ;
\node[honestnode] (v12) at (-0.5,-0.5) {}  ;
\node[honestnode] (v13) at (0.5,-0.5) {}  ;
\node[honestnode,label={below right:{\footnotesize $u$}}] (v14) at (0.5,0.5) {}  ;

\node[circle,minimum size=1.8cm,draw, dashed,label=right:$g_{i-1}$] (gi1) at (2,2)  {};
\node[honestnode] (v21) at (2-0.5,2.5) {}  ;
\node[honestnode] (v22) at (2-0.5,2-0.5) {}  ;
\node[honestnode] (v23) at (2.5,2-0.5) {}  ;
\node[honestnode] (v24) at (2.5,2.5) {}  ;

\node[circle,minimum size=1.8cm,draw, dashed,label=below:$g_{i+1}$] (gi2) at (4,0)  {};
\node[honestnode] (v31) at (4-0.5,0.5) {}  ;
\node[honestnode] (v32) at (4-0.5,-0.5) {}  ;
\node[honestnode] (v33) at (4.5,-0.5) {}  ;
\node[honestnode] (v34) at (4.5,0.5) {}  ;

\foreach \i in {1,2,3}{
  \path (v\i1) edge (v\i2);
  \path (v\i1) edge (v\i3);
  \path (v\i1) edge (v\i4); 
  \path (v\i2) edge (v\i3); 
  \path (v\i2) edge (v\i4); 
  \path (v\i3) edge (v\i4); 
}
\foreach \source/\dest   in {v14/v21,v14/v22,v14/v23,v14/v31,v14/v32,v14/v34}
    \path[-] (\source) edge[bend left=20] node[below, font=\sffamily\tiny] {} (\dest);

\end{tikzpicture}
\label{fig:example_graph_sf4}
}
\subfloat[Circle-based network]{
\begin{tikzpicture} 
[>=stealth',inner sep=0pt,  scale=0.8,
honestnode/.style={
circle, 
minimum size=3mm,
thick,
draw,
font=\sffamily
},every label/.style={thick}
]
\def \n {6}
\def \radius {1.5cm} 
\def \margin {45} 

\foreach \s in {1,2,3,4,5,6}
{
  \node[honestnode] (v\s) at ({(360/\n * (\s - 1))}:\radius) {};

}
\foreach \source/\dest in {v1/v3,v3/v5,v5/v1,v2/v4,v4/v6,v6/v2}
	\path[->] (\source) edge (\dest) ;
\foreach \source/\dest in {v1/v2,v2/v3,v3/v4,v4/v5,v5/v6,v6/v1}
	\path[<->] (\source) edge (\dest) ;

\end{tikzpicture}
\label{fig:example_graph_sf5}
}
\caption{{\small Examples of networks satisfying the 3-broadcasting property.}}
\label{fig:example_graphs}
\end{figure}
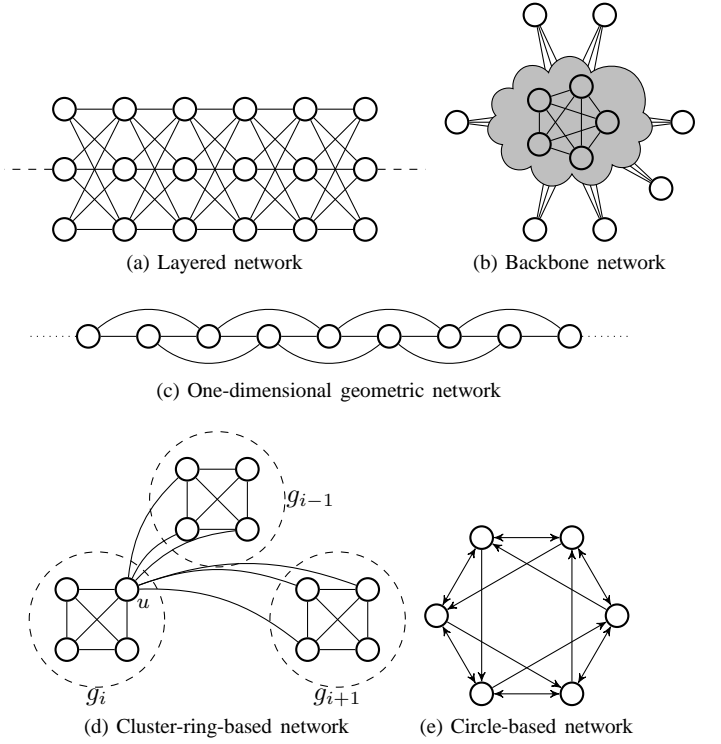  

\item
\medskip\noindent\textbf{Circle-based networks.}
Consider a graph $G_0\in\mathcal{G}_2$ of size $N>2k+1$ where each node $n$ has the set of consumers $\mathcal{S}_n$ which is
the output of the function: 
\begin{small}
\[
\begin{split}
 f\colon A &\to 2^A \\
 n &\mapsto \{(n-1)\bmod{N}, (n+1)\bmod{N},...(n+2k)\bmod{N}\}
\end{split}\]
 \end{small}
 where $A=\{0,1,2\ldots,(N-1)\}$.
The set of producers is $\mathcal{R}_n=\{(n+1)\bmod{N}, (n-1)\bmod{N}, (n-2)\bmod{N},\ldots,
(n-2k)\bmod{N}\}$ of size $2k+1$. 
Figure \ref{fig:example_graph_sf5} depicts an example of this graph with $N=6$ and $k=1$ in which the arrows express the direction of sending messages. 
For each source node, an (increased) ordering satisfying the  $m$-broadcasting-source property w.r.t. this source is the ordering of the nodes with respect to the distance from the source.                                                                                     

All protocols in  \cite{DBLP:journals/computing/BenkaouzE13,DBLP:conf/srds/BenkaouzGEH13,DBLP:conf/trustcom/EnglertG13,DBLP:conf/opodis/GuerraouiHKM09,DBLP:journals/jpdc/GuerraouiHKMV12} cannot be deployed in this graph since the condition  $\mathcal{S}_n\cap \mathcal{R}_n=\varnothing $ (proposed in these works) is not satisfied for every node $n$.

\end{enumerate}


\subsection{Graph construction}
\label{appxsub:algo_construct_m_broad}
We present briefly an algorithm for constructing an $m$-broadcasting graph from an arbitrary graph of size $N$ where the set of nodes $V=\{1,2,...,N\}$  as follows:

\begin{mdframed} 
\begin{enumerate}
 \item  Choose any ordering of the $N$ nodes.  
 \item  For $i=1$ to $m$: If node $i$ does not have an edge from the source, then insert an edge from the source to this node.
 \item For $i=m+1$ to $n$: Count the number of edges linking to node $i$ from nodes in the set $\{1,...,i-1\}$. If this count is less than $m$ then add edges from an arbitrary subset of nodes $\{1,...,i-1\}$ to node $m$ such that the count reaches $m$. 
\end{enumerate}
\end{mdframed}

The resulting graph will satisfy the $m$-broadcasting property.
\end{document}